\documentclass[acmsmall]{acmart}
\acmJournal{PACMPL}
\acmVolume{1}
\acmNumber{CONF} % CONF = POPL or ICFP or OOPSLA
\acmArticle{1}
\acmYear{2018}
\acmMonth{1}
\acmDOI{} % \acmDOI{10.1145/nnnnnnn.nnnnnnn}
\startPage{1}

\setcopyright{none}
\usepackage{booktabs}   %% For formal tables:
\usepackage{subcaption} %% For complex figures with subfigures/subcaptions
\usepackage{amsmath, bm, bbm}
\usepackage{stmaryrd}
\usepackage{amsthm}
\usepackage{hyperref}
\usepackage{graphicx}
\usepackage[normalem]{ulem}
\usepackage{cancel}
\usepackage{amsmath}
\usepackage[capitalize]{cleveref}
\usepackage[textsize=tiny]{todonotes}
\usepackage{wrapfig}
\usepackage{longfbox}

\usepackage{multirow}
\usepackage{tablefootnote}

\usepackage{pgfplots}
\usetikzlibrary{automata, calc}

\newcommand{\setdef}[2]{\left\{ #1\, \middle|\, \begin{matrix} #2 \end{matrix} \right\}}
\newcommand{\interpret}[1]{{\left\llbracket#1\right\rrbracket}}

\def\probbranch{\texttt{Pb}}
\def\stskip{skip}
\def\assume{\texttt{assume}}

\newcommand{\nextFullProb}[2][]{^{#1}\left[#2\right]}
\newcommand{\prob}[1][]{\mathbb{P}_{#1}\nextFullProb}
\newcommand{\bs}[1]{{\left\{#1\right\}}}

\def\dom{{\text{dom}}}

\def\scriptA{\mathcal{A}}

\def\scriptF{\mathcal{F}}

\def\scriptL{\mathcal{L}}
\def\scriptM{\mathcal{M}}

\def\scriptS{\mathcal{S}}

\def\doubleP{\mathbb{P}}

\def\mergeable{{\texttt{Mergeable}}}

\def\lbranch{{\texttt{L}}}
\def\rbranch{{\texttt{R}}}

\newcommand{\oriLprobbranch}[1][i]{{\probbranch_{(#1, \lbranch)}}}
\newcommand{\oriRprobbranch}[1][i]{{\probbranch_{(#1, \rbranch)}}}
\newcommand{\oriDprobbranch}[1][i]{{\probbranch_{(#1, d)}}}
\def\lprobbranch{\oriLprobbranch}
\def\rprobbranch{\oriRprobbranch}
\def\dprobbranch{\oriDprobbranch}
\newcommand{\stwhile}[2]{\textbf{while } #1 \textbf{ do} #2 \textbf{done}}
\def\tTrue{\texttt{True}}
\def\tFalse{\texttt{False}}

\newcommand{\zap}[1]{}

\crefname{fact}{Fact}{Facts}
\crefname{remark}{Remark}{Remarks}

\newcommand{\cbx}[1]{
  \lfbox[ 
    rounded,
    background-color=gray!25,
    border-color=gray!25,
    border-radius=4pt,
    height-align=middle,
    height=5.5pt,
  ]{#1}
}

\makeatletter
\newdimen\@tempdimd
\makeatother

\begin{document}

\newtheorem{remark}[theorem]{Remark}
\newtheorem{fact}[theorem]{Fact}

%% Title information
\title[\textsc{ProbTA}]{\textsc{ProbTA} : A Sound and Complete Proof Rule for Probabilistic Verification}         %% [Short Title] is optional;
                                        %% when present, will be used in
                                        %% header instead of Full Title.
%\titlenote{with title note}             %% \titlenote is optional;
                                        %% can be repeated if necessary;
                                        %% contents suppressed with 'anonymous'

%% Author information
%% Contents and number of authors suppressed with 'anonymous'.
%% Each author should be introduced by \author, followed by
%% \authornote (optional), \orcid (optional), \affiliation, and
%% \email.
%% An author may have multiple affiliations and/or emails; repeat the
%% appropriate command.
%% Many elements are not rendered, but should be provided for metadata
%% extraction tools.

%% Author with single affiliation.
\author{Guanyan Li}
%\authornote{with author1 note}          %% \authornote is optional;
                                        %% can be repeated if necessary
\orcid{0000-0002-4163-1840}             %% \orcid is optional
\affiliation{
  \department{Computer Science}              %% \department is recommended
  \institution{Tsinghua University}            %% \institution is required
  \city{Beijing}
  \country{China}                    %% \country is recommended
}
\affiliation{
  \department{Computer Science}             %% \department is recommended
  \institution{University of Oxford}           %% \institution is required
  \city{Oxford}
  \country{United Kingdom}                   %% \country is recommended
}
\email{guanyan.li@cs.ox.ac.uk}          %% \email is recommended

%% Author with two affiliations and emails.
\author{Zhilei Han}
% \authornote{with author2 note}          %% \authornote is optional;
                                        %% can be repeated if necessary
\orcid{0000-0001-9171-4997}             %% \orcid is optional
\affiliation{
  \department{School of Software}             %% \department is recommended
  \institution{Tsinghua University}           %% \institution is required
  \city{Beijing}
  \country{China}                   %% \country is recommended
}
\email{hzl21@mails.tsinghua.edu.cn}         %% \email is recommended

\author{Fei He}
\affiliation{
  \position{Associate Professor}
  \department{School of Software}             %% \department is recommended
  \institution{Tsinghua University}           %% \institution is required
  \city{Beijing}
  \country{China}                   %% \country is recommended
}
\email{hefei@tsinghua.edu.cn}

%% Abstract
%% Note: \begin{abstract}...\end{abstract} environment must come
%% before \maketitle command
\begin{abstract}
We propose a sound and complete proof rule \textsc{ProbTA} for quantitative analysis of violation probability of probabilistic programs.
Our approach extends the technique of trace abstraction with probability in the control-flow randomness style, in contrast to previous work of combining trace abstraction and probabilisitic verification which adopts the data randomness style.
In our method, a program specification is proved or disproved by
decomposing the program into different modules of traces.
Precise quantitative analysis is enabled by novel models proposed to bridge program verification and probability theory. 
%and further analysis on them.
Based on the proof rule, we propose a new automated algorithm via CEGAR involving multiple technical issues unprecedented in non-probabilistic trace abstraction and data randomness-based approach.
\end{abstract}

%% 2012 ACM Computing Classification System (CSS) concepts
%% Generate at 'http://dl.acm.org/ccs/ccs.cfm'.

%% End of generated code

%% Keywords
%% comma separated list
  %% \keywords are mandatory in final camera-ready submission

%% \maketitle
%% Note: \maketitle command must come after title commands, author
%% commands, abstract environment, Computing Classification System
%% environment and commands, and keywords command.
\maketitle

\section{Introduction}

\label{sec: intro}
% !TEX root = ./main.tex

Since the early days of computer science, formalisms for reasoning about probability have been widely studied. 
Extending classical imperative programs, probabilistic programs enable efficient solution to many algorithmic problems~\cite{efficiency_of_probabilistic_algorithm_1,efficiency_of_probabilistic_algorithm_2,efficiency_of_probabilistic_algorithm_3}.
It has also led to new models which play key roles in cryptography~\cite{encrption_using_probability}, linguistics~\cite{plcfrs} and especially machine learning where generic models \cite{ppl_anglican,ppl_church,ppl_stan} are expressed thereby.

In this paper, we revisit the verification problem of probabilistic programs. More specifically, we focus on quantitative analysis of violation probability of probabilistic programs~\cite{exponential_analysis_of_probabilistic_program,trace_abstraction_modulo_probability} which can be described as follows: given a probabilistic program $P$, a precondition $\varphi_e$, a postcondition $\varphi_f$ and a threshold $\beta$, answer the question that whether the probability of violating the Hoare-triple $\bs{\varphi_e}\ P\ \bs{\varphi_f}$ is no greater than $\beta$. Following the notation of \cite{trace_abstraction_modulo_probability}, we denote this by: 
$$ \vdash_\beta \bs{\varphi_e}\ P\ \bs{\varphi_f} $$

One of the distinct features of our approach is the combination of probabilistic program verification with trace abstraction \cite{trace_abstraction}, a technique that has been successfully employed in wide areas of non-probabilistic program verification and analysis \cite{site_trace_abstraction,trace_abstraction_termination_cav14,trace_abstraction_incremental,trace_abstraction_termination_incremental}.
Trace abstraction is a \emph{decomposition}-based technique in that if one can decompose a program, regarded as an automaton accepting control flow traces, into a set of certified automata whose accepting traces are guaranteed to be non-violating, then one is able to declare the non-violation of the program against the given specification.

The combination of probabilistic verification and trace abstraction has been pioneered by Smith et al.
In their recent work~\cite{trace_abstraction_modulo_probability}, the authors wisely handle the effects of probability within a trace by assigning probabilities to program expressions,
and utilising program synthesis techniques to model the effects of probabilistic distributions which may affect the validity of the Hoare-triple.
Despite the ability of their method to handle quite complex and interesting examples, its incompleteness is obvious.
As is discussed in~\cite{trace_abstraction_modulo_probability}, the automated algorithm proposed cannot handle the case $\bs{\texttt{True}}\ X \sim Flip(0.5); Y \sim Flip(0.5) \bs{X \wedge Y}$,
where $1$ is always returned as the upper bound.
Although optimisations are available, a fundamental improvement is still hard to obtain.

The approach of~\cite{trace_abstraction_modulo_probability} is classified as an instance of \emph{data randomness} in a recent work by Wang et al.~\cite{pmaf}, where the consequence of such methods is analysed more thoroughly.
In this work, however, we are interested in another edge of the spectrum: the \emph{control-flow randomness} introduced in~\cite{pmaf},
which provides a new trace-based perspective of programs where probabilities are considered on the level of statements rather than expressions, 
and the effects of probabilistic choices are modeled directly on the control flow. This also enables \emph{non-determinism}, an important feature in the context of probabilistic verification.

Nevertheless, it brings a major challenge we need to address: how to compute the probability of a set of traces (against some specification)? In a data randomness-based approach like~\cite{trace_abstraction_modulo_probability}, one just needs to sum up the probability of all traces without considering interdependency, since information required for probability computation is fully presented on each trace. 
In our case, however, every trace only contains partial information, and probability computation depends heavily on the structure used to represent the set of traces. Since these structures are generally not unique, we must be able to compute probability regardless of the underlying structure. As a consequence, a fundamentally different proof technique is required.

To tackle this problem, we propose \emph{control-flow Markov decision process (CFMDP)} and \emph{control-flow Markov chain (CFMC)}, two natural extensions of the traditional probabilistic models MDP and MC with control flow information, which can be used to represent the trace set.
These models serve as bridges between program verification and the classical theory of stochastic process, enabling precise quantitative analysis by using existing techniques.
By proper supportive definitions of probabilities from the view of traces,
we show that \emph{normalised} (in a sense that is explained in \cref{sec: the proof rule}) CFMDP enables probability computation of trace sets.
Based on these models, we propose a new sound and complete \emph{proof rule} called \textsc{ProbTA} for probabilistic program verification, extending trace abstraction with probability in a distinct manner. 

Furthermore, we present an automated algorithm based on \emph{Counterexample-guided Abstraction Refinement(CEGAR)} to apply our new proof rule.
The algorithm naturally extends the non-probabilistic CEGAR-based trace abstraction with another routine to tackle violating but tolerable traces.
In the algorithm, we address multiple technical issues encountered in neither non-probabilistic trace abstraction nor data randomness-based approach,
such as generalisation of probable traces and probabilistic counterexample extractions.
Given all these, we provide a relatively comprehensive framework of probabilistic trace abstraction with probability in the control-flow randomness style.

To sum up, our main contributions are:

\begin{itemize}
	\item We present CFMDP and CFMC, two natural extensions of existing models equipped with control flow information. We show the models enable reasoning of probabilistic behaviours in a control-flow randomness-based perspective. 
	\item We propose a sound and complete proof rule \textsc{ProbTA} which
        extends the traditional trace abstraction with probaility. %and prove its soundness and completeness.
	\item We present an automated algorithm making use of \textsc{ProbTA} via
        CEGAR, solving several interesting technical issues that are not previously encountered.
\end{itemize}

The rest of the paper is organised as follows.
After demonstrating the basic idea of our proof rule on a motivating example in \cref{sec: motivation example and overview}
and objects including CFMDP and CFMC and the problem to solve is formally presented in \cref{sec: preliminaries},
\cref{sec: the proof rule} presents our proof rules and proves its property,
followed by \cref{sec: algorithm} where an automated algorithm using the new proof rule is proposed.
Finally we survey and compare related works in \cref{sec: related work} and conclude in \cref{sec: conlusion}.

\section{Motivating Example}
\label{sec: motivation example and overview}

% !TEX root = ./main.tex

In this section, we illustrate our method by a motivating example followed by discussion.

\subsection{Motivating Example}

\begin{example}
	\label{example: motivation}
    Consider the following program, where $X$ and $C$ are integer variables and 
    $\oplus$ is the fair binary probabilistic choice operator:
	\begin{align*}
		1\quad & \bs{\texttt{True}} \\
		2\quad & X ::= 0; \\
		3\quad & C ::= 0 \oplus \stskip; \\
		4\quad & \stwhile{C > 0}{\\
		5\quad & \qquad X ::= X + 1 \oplus \stskip; \\
		6\quad & \qquad C ::= C - 1 \\
		7\quad & } \\
		8\quad & \bs{X = 0}
	\end{align*}
\end{example}

The program could be regarded as a coin flipping game between two players, Alice and Bob. 
They make an initial toss first (line 3), where if the coin is head-down, the game ends and Alice wins.
Otherwise, Bob decides a number (the variable $C$) of rounds to continue playing. 
Alice continues guessing no head-up occurs during the additional rounds. 
As the actual number of head-up(s) is kept in the variable $X$ (line 5), the postcondition $X = 0$ means Alice wins the game.
An interesting property of this example is that the probability of Bob winning the game, corresponding to the violation of the postcondition, increases as the initial value of $C$ increases. But the upper bound of the violation probability, $0.5$, is only attained when $C \to \infty$. 

We consider the following question about the example: does the postcondition hold up to a tolerant probability $\beta := 0.3$ 
(is the probability of Bob winning the game no greater than $0.3$)? 
If not, what is a possible counterexample 
(what is a possible initial value of $C$ such that the probability of Bob winning the game is greater than $0.3$, and under this circumstance, what is a possible playing sequence leading to violation)?

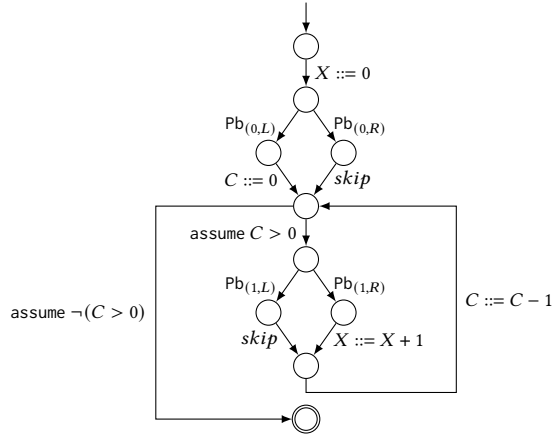
\begin{figure}[h]
	\begin{tikzpicture}[yscale=0.7]
    % place nodes
    \node at (0, 1) (N0) {};
    \node[draw, circle] at (0, 0) (N1) {};
    \node[draw, circle] at (0, -1) (N2) {};
    \node[draw, circle] at (0.5, -2) (N3) {};
    \node[draw, circle] at (-0.5, -2) (N4) {};
    \node[draw, circle] at (0, -3) (N5) {};
    \node[draw, circle] at (0, -4) (N6) {};
    \node[draw, circle] at (-0.5, -5) (N7) {};
    \node[draw, circle] at (0.5, -5) (N8) {};
    \node[draw, circle] at (0, -6) (N9) {};
    \node[draw, circle, accepting] at (0, -7) (NE) {};

    % draw edges
    \draw[-latex] (N0) -- (N1);
    \draw[-latex] (N1) -- (N2)  node[midway, right] {\scriptsize $X::=0$};
    \draw[-latex] (N2) -- (N3)  node[midway, right] {\scriptsize $\probbranch_{(0,R)}$};
    \draw[-latex] (N2) -- (N4)  node[midway, left] {\scriptsize $\probbranch_{(0,L)}$};
    \draw[-latex] (N3) -- (N5)  node[midway, right] {\scriptsize $\stskip$};
    \draw[-latex] (N4) -- (N5)  node[midway, left] {\scriptsize $C::=0$};
    \draw[-latex] (N5) -- (N6)  node[midway, left] {\scriptsize $\assume\ C>0$};
    \draw[-latex] (N6) -- (N7)  node[midway, left] {\scriptsize $\probbranch_{(1,L)}$};
    \draw[-latex] (N6) -- (N8)  node[midway, right] {\scriptsize $\probbranch_{(1,R)}$};
    \draw[-latex] (N7) -- (N9)  node[midway, left] {\scriptsize $\stskip$};
    \draw[-latex] (N8) -- (N9)  node[midway, right] {\scriptsize $X::=X+1$};
    \draw[-latex] (N9) -- (0, -6.5) -- (2, -6.5) -- (2, -3) node[midway, right]  {\scriptsize $C ::= C - 1$} -- (N5);
    \draw[-latex] (N5) -- (-2, -3) -- (-2, -7) node[midway, left]  {\scriptsize $\assume\ \neg(C>0)$} -- (NE);
\end{tikzpicture}
	\caption{CFA $\mathcal{A}$ of \cref{example: motivation}}
	\label{figure: motivation example cfa}
\end{figure}

We first transform the program into a \emph{probabilistic control-flow automaton} (abbreviated below as PCFA or simply CFA) $\mathcal{A}$, which is depicted in \cref{figure: motivation example cfa}. 
The translation is a standard routine~\cite{ast_by_omega_decompose} with special handling of probabilistic choice statements: each branch of the probabilistic choice is represented by $\probbranch_{(i, D)}$, where $i$ is the identifier syntacically assigned to each probabilistic choice statement $s_1 \oplus s_2$, such as $0$ and $1$ for the statements in line 3 and 5 respectively, while $D$ distinguish the two branches of the $\oplus$ operator by $\texttt{L}$ (for Left) or $\texttt{R}$ (for Right).

% \begin{figure}
% 	\includegraphics[width=3cm]{example_graph}
% 	\caption{CFA of \cref{example: motivation}}
% 	\label{figure_old: motivation example cfa}
% \end{figure}

Then, we proceed our proof by showing the traces of the program, or equivalently the language of $\mathcal{A}$, are covered by the three automata $\mathcal{A}_1, \mathcal{A}_2$ and $\mathcal{A}_3$ presented\footnote{The set of all labels on \cref{figure: motivation example cfa} is denoted by $\Sigma$.} in \cref{figure: motivation example category 1}, \cref{figure: motivation example category 2} and \cref{figure: motivation example violation automaton}, so that $\scriptL(\mathcal{A}) \subseteq \scriptL(\mathcal{A}_1) \cup \scriptL(\mathcal{A}_2) \cup \scriptL(\mathcal{A}_3)$. 
These automata are a variant of interpolant automata~\cite{trace_abstraction} 
in which every state is labeled with a proposition and every transition with label $\sigma \in \Sigma$ entails the validity of the Hoare triple $\bs{\varphi_s}\ \sigma\ \bs{\varphi_t}$, where $\varphi_s$ and $\varphi_t$ are the labeled proposition of the source and target state respectively.

% \begin{figure}
% 	\includegraphics[width=\linewidth]{figs/proof_1_graph}
% 	\caption{Categ.~1: Safe}
% 	\label{figure_old: motivation example category 1}
% \end{figure}

%%\begin{figure}
%%	\input{figs/fig_categ1.tex}
%%	\caption{The automaton $\mathcal{A}_1$}
%%	\label{figure: motivation example category 1}
%%\end{figure}

% \begin{figure}
% 	\includegraphics[width=\linewidth]{figs/proof_2_graph}
% 	\caption{Categ.~2: Safe}
% 	\label{figure_old: motivation example category 2}
% \end{figure}

%\begin{figure}
%	\input{figs/fig_categ2.tex}
%	\caption{The automaton $\mathcal{A}_2$}
%	\label{figure: motivation example category 2}
%\end{figure}

% \begin{figure}
% 	\includegraphics[width=\linewidth]{violation_automaton_graph}
% 	\caption{Categ.~3: Violation}
% 	\label{figure_old: motivation example violation automaton}
% \end{figure}

%\begin{figure}
%	\input{figs/fig_violation.tex}
%	\caption{The automaton $\mathcal{A}_3$}
%	\label{figure: motivation example violation automaton}
%\end{figure}

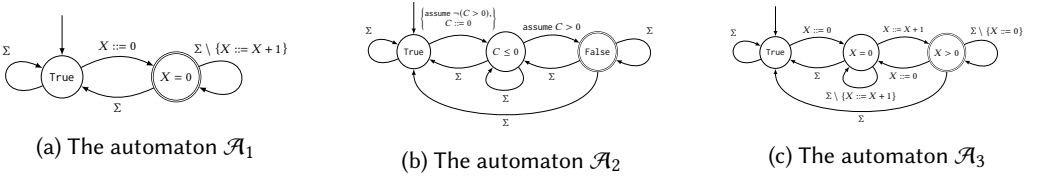
\begin{figure*}[t]
	\centering
	\begin{subfigure}{0.3\textwidth}
		\centering
		\resizebox{\textwidth}{!}{
		\begin{tikzpicture}[xscale=0.75]
    % place nodes
    \node[draw, circle, minimum size=20pt] at (0, 0) (N0) {\scriptsize \tTrue};
    \node[draw, circle, accepting, minimum size=20pt] at (3, 0) (N1) {\scriptsize $X=0$};
    % \node[draw, circle] at (4, 0) (N2) {};
    % \node[draw, circle] at (6, 0) (N3) {};

    % \node[draw, circle] at (4, -1.5) (N4) {};
    % \node[draw, circle] at (4, -3) (N5) {};
    % \node[draw, circle] at (4, -4.5) (N6) {};
    % \node[draw, circle, accepting] at (1.5, -3) (N7) {};

    % \node[draw, circle] at (0, -3) (N8) {};
    % \node[draw, circle] at (6, -3) (N9) {};

    % draw edges
    \draw[-latex] (N0) edge [bend left=20] node[above] {\scriptsize $X ::= 0$} (N1);
    \draw[-latex] (N1) edge [bend left=20] node[below] {\scriptsize $\Sigma$} (N0);
    \draw[-latex] (N0) edge [out=200, in=160,looseness=12] node[above=0.3] {\scriptsize $\Sigma$} (N0); 
    \draw[-latex] (N1) edge [out=20, in=-20,looseness=12] node[above=0.3] {\scriptsize $\Sigma \setminus \bs{X ::= X + 1}$} (N1);
    \draw[-latex] (0, 1.2) -- (N0);
    % \draw[-latex] (N0) -- (N1)  node[midway, above] {\scriptsize $X::=0;\probbranch_{(0,R)};\stskip$};
    % \draw[-latex] (N1) -- (N2)  node[midway, above] {\scriptsize $\assume\ C>0$};
    % \draw[-latex] (N2) -- (N3)  node[midway, above] {\scriptsize $\probbranch_{(1,R)};\stskip$};
    % \draw[-latex] (N3) -- (6, 1) -- (2, 1) node[midway, above] {\scriptsize $C::=C-1$} -- (N1);
    % \draw[-latex] (N2) -- (N4)  node[midway, left] {\scriptsize $\probbranch_{(1,L)};X::=X-1$};
    % \draw[-latex] (N4) -- (N5)  node[midway, right] {\scriptsize $C::=C-1$};

    % \draw[-latex] (N5) -- (N6)  node[midway, left] {\scriptsize $\assume\ C>0$};
    % \draw[-latex] (N5) -- (N7)  node[midway, above] {\scriptsize $\assume\ \neg(C>0)$};

    % \draw[-latex] (N6) -- (0,-4.5)  node[midway, above] {\scriptsize $\probbranch_{(1,L)}$} -- (N8);
    % \draw[-latex] (N8) -- (0, -1.5) -- (N4) node[midway, below] {\scriptsize $X::=X+1$};
    % \draw[-latex] (N6) -- (6,-4.5)  node[midway, above] {\scriptsize $\probbranch_{(1,R)}$} -- (N9);
    % \draw[-latex] (N9) -- (6,-1.5) -- (N4) node[midway, above] {\scriptsize $\stskip$};
\end{tikzpicture}
		}
		\caption{The automaton $\mathcal{A}_1$}
		\label{figure: motivation example category 1}
	\end{subfigure}
	\hfill
	\begin{subfigure}{.3\textwidth}
		\resizebox{\textwidth}{!}{
		\begin{tikzpicture}[xscale=0.75]
    % place nodes
    \node[draw, circle, minimum size=20pt] at (0, 0) (N0) {\scriptsize \tTrue};
    \node[draw, circle, accepting, minimum size=20pt ] at (6, 0) (N1) {\scriptsize \tFalse};
    \node[draw, circle, minimum size=20pt] at (3, 0) (N2) {\scriptsize $C\le 0$};
    % \node[draw, circle] at (4, 0) (N2) {};
    % \node[draw, circle] at (6, 0) (N3) {};

    % \node[draw, circle] at (4, -1.5) (N4) {};
    % \node[draw, circle] at (4, -3) (N5) {};
    % \node[draw, circle] at (4, -4.5) (N6) {};
    % \node[draw, circle, accepting] at (1.5, -3) (N7) {};

    % \node[draw, circle] at (0, -3) (N8) {};
    % \node[draw, circle] at (6, -3) (N9) {};

    % draw edges
    % \draw[-latex] (N0) edge [bend left=20] node[above] {\scriptsize $\{$True$\}st\{X=0\}$} (N1);
    \draw[-latex] (N1) edge [out=-90, in=-90, looseness=0.6] node[below] {\scriptsize $\Sigma$} (N0);
    \draw[-latex] (N2) edge [out=-45, in=-135, looseness=3] node[below] {\scriptsize $\Sigma$} (N2);
    \draw[-latex] (N0) edge [bend left=20] node[above] {\tiny $\bs{\begin{matrix}\assume\ \neg (C > 0), \\ C ::= 0\end{matrix}}$} (N2);
    \draw[-latex] (N2) edge [bend left=20] node[below] {\scriptsize $\Sigma$} (N0);
    \draw[-latex] (N2) edge [bend left=20] node[above] {\scriptsize $\assume\ C > 0$} (N1);
    \draw[-latex] (N1) edge [bend left=20] node[below] {\scriptsize $\Sigma$} (N2);
    \draw[-latex] (N0) edge [out=200, in=160,looseness=12] node[above=0.3] {\scriptsize $\Sigma$} (N0);
    \draw[-latex] (N1) edge [out=20, in=-20,looseness=12] node[above=0.3] {\scriptsize $\Sigma$} (N1);
    \draw[-latex] (0, 1.2) -- (N0);
    % \draw[-latex] (N0) -- (N1)  node[midway, above] {\scriptsize $X::=0;\probbranch_{(0,R)};\stskip$};
    % \draw[-latex] (N1) -- (N2)  node[midway, above] {\scriptsize $\assume\ C>0$};
    % \draw[-latex] (N2) -- (N3)  node[midway, above] {\scriptsize $\probbranch_{(1,R)};\stskip$};
    % \draw[-latex] (N3) -- (6, 1) -- (2, 1) node[midway, above] {\scriptsize $C::=C-1$} -- (N1);
    % \draw[-latex] (N2) -- (N4)  node[midway, left] {\scriptsize $\probbranch_{(1,L)};X::=X-1$};
    % \draw[-latex] (N4) -- (N5)  node[midway, right] {\scriptsize $C::=C-1$};

    % \draw[-latex] (N5) -- (N6)  node[midway, left] {\scriptsize $\assume\ C>0$};
    % \draw[-latex] (N5) -- (N7)  node[midway, above] {\scriptsize $\assume\ \neg(C>0)$};

    % \draw[-latex] (N6) -- (0,-4.5)  node[midway, above] {\scriptsize $\probbranch_{(1,L)}$} -- (N8);
    % \draw[-latex] (N8) -- (0, -1.5) -- (N4) node[midway, below] {\scriptsize $X::=X+1$};
    % \draw[-latex] (N6) -- (6,-4.5)  node[midway, above] {\scriptsize $\probbranch_{(1,R)}$} -- (N9);
    % \draw[-latex] (N9) -- (6,-1.5) -- (N4) node[midway, above] {\scriptsize $\stskip$};
\end{tikzpicture}
		}
		\caption{The automaton $\mathcal{A}_2$}
		\label{figure: motivation example category 2}
	\end{subfigure}
	\hfill
	\begin{subfigure}{.3\textwidth}
		\resizebox{\textwidth}{!}{
		\begin{tikzpicture}[xscale=0.75]
    % place nodes
    \node[draw, circle, minimum size=20pt] at (0, 0) (N0) {\scriptsize \tTrue};
    \node[draw, circle, accepting, minimum size=20pt ] at (6, 0) (N1) {\scriptsize $X > 0$};
    \node[draw, circle, minimum size=20pt] at (3, 0) (N2) {\scriptsize $X = 0$};
%    \node[draw, circle] at (0, 0) (N0) {};
%    \node[draw, circle] at (2, 0) (N1) {};
%    \node[draw, circle] at (4, 0) (N2) {};
%    \node[draw, circle, accepting] at (6.5, 0) (N3) {};

%    \node[draw, circle] at (4, -1.5) (N4) {};
%    \node[draw, circle] at (4, -3) (N5) {};
%    \node[draw, circle] at (4, -4.5) (N6) {};
%    \node[draw, circle, accepting] at (1.5, -3) (N7) {};
%
%    \node[draw, circle] at (0, -3) (N8) {};
%    \node[draw, circle] at (6, -3) (N9) {};

    % draw edges
    \draw[-latex] (N1) edge [out=-90, in=-90, looseness=0.6] node[below] {\scriptsize $\Sigma$} (N0);
    \draw[-latex] (N0) edge [bend left=20] node[above] {\scriptsize $X ::= 0$} (N2);
    \draw[-latex] (N2) edge [out=-45, in=-135, looseness=3] node[below] {\scriptsize $\Sigma \setminus \bs{X ::= X + 1}$} (N2);
%    \draw[-latex, very thick, dash dot] (N2) edge [bend left=20] node[above] {\tiny $\Sigma$} (N1);
    \draw[-latex] (N2) edge [bend left=20] node[below] {\scriptsize $\Sigma$} (N0);
    \draw[-latex] (N2) edge [bend left=20] node[above] {\tiny $X ::= X + 1$} (N1);
    \draw[-latex] (N1) edge [bend left=20] node[below] {\scriptsize $X ::= 0$} (N2);
    \draw[-latex] (N0) edge [out=200, in=160,looseness=12] node[above=0.3] {\scriptsize $\Sigma$} (N0);
    \draw[-latex] (N1) edge [out=20, in=-20,looseness=12] node[above=0.3] {\scriptsize $\Sigma \setminus \bs{X ::= 0}$} (N1);
    \draw[-latex] (0, 1.2) -- (N0);
\end{tikzpicture}
		}
		\caption{The automaton $\mathcal{A}_3$}
		\label{figure: motivation example violation automaton}
	\end{subfigure}
	\caption{a decompostion of $\mathcal{A}$}
\end{figure*}

Essentially, $\mathcal{A}_1$ and $\mathcal{A}_2$ represent a subset of program traces which share common proofs for non-violation. For example, the automaton $\mathcal{A}_1$ represent traces where $X$ is ultimately 0, while $\mathcal{A}_2$ represent all infeasible traces since the accepting state has proposition $\texttt{False}$. $\mathcal{A}_1$ and $\mathcal{A}_2$ thus accepts exactly all the non-violating traces.

The automaton $\mathcal{A}_3$, however, is an over-approximation of the violating traces. To answer the above question, we need to analyse $\mathcal{A}_3$ more thoroughly.
The analysis starts by removing from $\mathcal{A}_3$ redundant traces that are not in $\mathcal{A}$, as well as traces that have been proved to be non-violating by $\mathcal{A}_1$ and $\mathcal{A}_2$ (note that they are not necessarily disjoint). This step is equivalent to the automata-theoretic operation $(\mathcal{A}_3 \setminus (\mathcal{A}_1 \cup \mathcal{A}_2)) \cap \mathcal{A}$.
The trace set is then representable by the regular expression
$\pi = \rho_{\texttt{Pre}}(\rho_{\texttt{Skip}} | \rho_{\texttt{Inc}})^*\rho_{\texttt{Inc}}(\rho_{\texttt{Skip}} | \rho_{\texttt{Inc}})^*$
where $\rho_{\texttt{Pre}}$ is the trace \cbx{$X ::= 0$}\cbx{$\rprobbranch[0]$}\cbx{$\stskip$}, 
	  $\rho_{\texttt{Inc}}$ is \cbx{$\assume\ C > 0$} \cbx{$ \rprobbranch[1]$} \cbx{$ X ::= X + 1$} \cbx{$ C ::= C - 1$}, and 
	  $\rho_{\texttt{Skip}}$ is \cbx{$\assume\ C > 0$} \cbx{$ \lprobbranch[1]$} \cbx{$ \stskip$} \cbx{$ C ::= C - 1$}.
Intuitively,
$\rho_{\texttt{Pre}}$ is the prefix that avoids resetting $C$ to $0$,
$\rho_{\texttt{Inc}}$ is the loop path that increases $X$ and 
$\rho_{\texttt{Skip}}$ is the loop path that keeps $X$ unchanged.
This captures exactly all violating traces of the example.

One could then be able to give a negative answer to the question, by picking 3 traces in $\pi$, displayed in \cref{figure: a counterexample of motivation example}, whose total violation probability is $0.375$ under the shared precondition $C = 2$.
Intuitively, if Bob let $C$ be $2$, he has a winning probability of $0.375$ and a possible playing sequence is: the coin is head-up in the initial toss, and in the following two tosses, the coin is head-up at least once.

% \begin{figure*}
% 	\includegraphics[width=\textwidth]{figs/motivation_example_counterexample_graph}
% 	\caption{A Counterexample of \cref{example: motivation}}
% 	\label{figure_old: a counterexample of motivation example}
% \end{figure*}

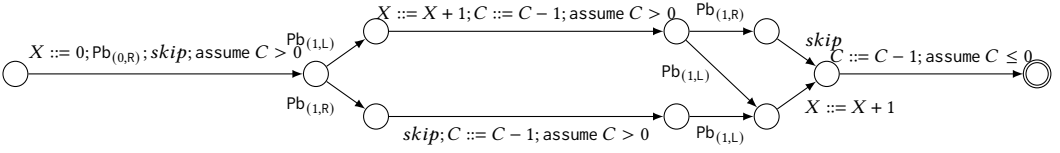
\begin{figure}
	\begin{tikzpicture}[yscale=0.7,scale=0.8]
    % place nodes
    \node[draw, circle] at (-3, 0) (N0) {};
    \node[draw, circle] at (2, 0) (N1) {};
    \node[draw, circle] at (3, 1) (N2) {};
    \node[draw, circle] at (8, 1) (N4) {};
    \node[draw, circle] at (9.5, 1) (N6) {};
    \node[draw, circle] at (3, -1) (N3) {};
    \node[draw, circle] at (8, -1) (N5) {};
    \node[draw, circle] at (9.5, -1) (N7) {};
    \node[draw, circle] at (10.5, 0) (N8) {};
    \node[draw, accepting, circle] at (14, 0) (N9) {};

    % \node[draw, circle] at (4, 0) (N2) {};
    % \node[draw, circle] at (6, 0) (N3) {};

    % \node[draw, circle] at (4, -1.5) (N4) {};
    % \node[draw, circle] at (4, -3) (N5) {};
    % \node[draw, circle] at (4, -4.5) (N6) {};
    % \node[draw, circle, accepting] at (1.5, -3) (N7) {};

    % \node[draw, circle] at (0, -3) (N8) {};
    % \node[draw, circle] at (6, -3) (N9) {};

    % draw edges
    \draw[-latex] (N0) -- (N1)  node[midway, above] {\scriptsize $X::=0; \rprobbranch[0];\stskip;\assume\ C>0$};
    \draw[-latex] (N1) --  node[pos=0.9, left=0.2] {\scriptsize $\lprobbranch[1]$} (N2);
    \draw[-latex] (N1) --  node[pos=0.9, left=0.2] {\scriptsize $\rprobbranch[1]$} (N3);
    \draw[-latex] (N2) -- node[above] {\scriptsize $X::=X+1;C::=C-1;\assume\ C>0$} (N4);
    \draw[-latex] (N3) -- node[below] {\scriptsize $\stskip;C::=C-1;\assume\ C>0$} (N5);
    \draw[-latex] (N4) -- node[above] {\scriptsize $\rprobbranch[1]$} (N6);
    \draw[-latex] (N5) -- node[below] {\scriptsize $\lprobbranch[1]$} (N7);
    \draw[-latex] (N4) -- node[left] {\scriptsize $\lprobbranch[1]$} (N7);
    \draw[-latex] (N6) -- node[pos=0.1, right=0.2] {\scriptsize $\stskip$} (N8);
    \draw[-latex] (N7) -- node[pos=0.1, right=0.2] {\scriptsize $X::=X+1$} (N8);
    \draw[-latex] (N8) -- node[above] {\scriptsize $C::=C-1;\assume\ C\le 0$} (N9);
    % \draw[-latex] (N1) -- (N2)  node[midway, above] {\scriptsize $\assume\ C>0$};
    % \draw[-latex] (N2) -- (N3)  node[midway, above] {\scriptsize $\probbranch_{(1,R)};\stskip$};
    % \draw[-latex] (N3) -- (6, 1) -- (2, 1) node[midway, above] {\scriptsize $C::=C-1$} -- (N1);
    % \draw[-latex] (N2) -- (N4)  node[midway, left] {\scriptsize $\lprobbranch[1];X::=X-1$};
    % \draw[-latex] (N4) -- (N5)  node[midway, right] {\scriptsize $C::=C-1$};

    % \draw[-latex] (N5) -- (N6)  node[midway, left] {\scriptsize $\assume\ C>0$};
    % \draw[-latex] (N5) -- (N7)  node[midway, above] {\scriptsize $\assume\ \neg(C>0)$};

    % \draw[-latex] (N6) -- (0,-4.5)  node[midway, above] {\scriptsize $\lprobbranch[1]$} -- (N8);
    % \draw[-latex] (N8) -- (0, -1.5) -- (N4) node[midway, below] {\scriptsize $X::=X+1$};
    % \draw[-latex] (N6) -- (6,-4.5)  node[midway, above] {\scriptsize $\probbranch_{(1,R)}$} -- (N9);
    % \draw[-latex] (N9) -- (6,-1.5) -- (N4) node[midway, above] {\scriptsize $\stskip$};
\end{tikzpicture}
	\caption{A counterexample of \cref{example: motivation}}
	\label{figure: a counterexample of motivation example}
\end{figure}

\subsection{Discussion}
\label{subsec: overview}

\cref{example: motivation} presents some basic ideas that underpins our proof rule in \cref{sec: the proof rule}.
We prove (non-)violation by decomposing all traces of the program into violating and non-violating modules. The non-violating modules contain traces that are guaranteed to be non-violating, while violating modules are over-approximation of violating trace set.
We then analyse the violating modules, either prove the upper bound of violation probability in these modules is no greater than the threshold $\beta$, or find a set of violating traces with violation probability beyond $\beta$.  

Nevertheless, soundness of the proof technique is non-trivial. 
In general, a violating module does not only contain traces that lead to violation,
since the set of violating traces is generally not expressible in regular language.
Besides, a violating module may have a different structure from the original CFA and thus, as is stated in \cref{sec: intro}, probability computation can not rely on the representation of the trace sets.
All these contribute to the hardness of proper analysis of violating modules.

In \cref{example: motivation}, one can easily find the counterexample by enumerating a few violating traces. The above dilemma, however, can be partially seen when we set $\beta := 0.5$.
In this case, we need to show the non-violation of program traces in $\mathcal{A}_3$, that is, the violation probability is no greater than $0.5$.
There are two major problems we need to address: 
1) how to compute the upper bound of violation probability to show the non-violation?
2) why can we declare the non-violation of $\mathcal{A}_3$ guarantees non-violation of $\mathcal{A}$?
The key towards solution lies in a reasonable supportive definition of probability of a set of traces. A more thorough discussion is given below in \cref{subsec: pre-proof rules,subsec: probabilities from trace view}.

\section{Preliminaries}
\label{sec: preliminaries}

% !TEX root = ./main.tex
In this section, we formally define the problem to address in this paper along with some definitions.

\paragraph{Measure and Probability}
Let $\Omega$ be a non-empty set.
A $\sigma$-algebra $\scriptF$ is a class of special subsets of $\Omega$, such that: 
$\Omega \in \scriptF$ and $\scriptF$ is closed under complement and countable unions.
A set $S$ is called $\scriptF$-measurable if $S \in \scriptF$.
A \emph{measurable space} is given by a pair $(\Omega, \scriptF)$ where $\Omega$ is non-empty and $\scriptF$ is a $\sigma$-algebra on $\Omega$.
A \emph{probability space} is then a triple $(\Omega, \scriptF, \doubleP)$ 
where $(\Omega, \scriptF)$ is a measurable space and 
	  $\doubleP$ is the probability measure that assigns every $\scriptF$-measurable set a number between $0$ and $1$ such that: 
	  $\prob{\Omega} = 1$ and 
	  for any countable sequence of pairwise-disjoint $\scriptF$-measurable sets $S_1, S_2, ...$, $\prob{\biguplus_{i = 1}^\infty S_i} = \sum_{i = 1}^\infty \prob{S_i}$.

\paragraph{Program}
Throughout the paper, we refer to \emph{programs} as classical imperative programs\footnote{For generality and brevity, we do not introduce a specific underlying language since it is irrelevant for discussion.} equipped with fair binary probabilistic choice $\oplus$ and binary non-deterministic operator $\circledast$ \emph{on statements}.
Our framework is compatible with the more general version of probabilistic choices $\oplus_p$.
Notwithstanding, it is well-known that fair binary probabilistic choice is sufficient if the underlying programming language is Turing complete~\cite{prob_turing_complete}.
We adopt the standard definition of a \emph{program state} to be a valuation of program variables.

\paragraph{Data Randomness vs.~Control-Flow Randomness}
This naming for the two existing styles of extending programs with probabilities in the literature is proposed by Wang et al.~\cite{pmaf}.
Data randomness style programs model probabilities via \emph{sampling} expressions, which implicitly introduces probabilities by program semantics,
while control-flow randomness style programs model probabilities by statements or the control-flow, which explicitly introduces probabilities by program syntax.
For example,
the statement in data randomness style \cbx{$X \sim Flip(0.5)$} is equivalent to \cbx{$X ::= 0 \oplus X ::= 1$} in control-flow randomness style.

\paragraph{Specification}
A specification for a program is a pair of propositions $(\varphi_e, \varphi_f)$ over program variables, referred to as pre- and post-condition
respectively. 
We use $\Phi$ to denote the pair.

\paragraph{PCFA}
The main focus of this paper is with \emph{probabilistic control-flow
automaton}, which is abbreviated as PCFA or just CFA below. It is essentially
a non-probabilistic CFA with probabilistic and non-deterministic labels. 

Each non-random statement of the program is regarded as a label.
We assume each of the probabilistic or non-deterministic operators in the program has a
unique identifier, say $i$. 
We use 
%\mytodo{symbol conflict with program $P$}{
$\lprobbranch$ and
$\rprobbranch$
%}
to label 
the two branches of the probabilistic operator $i$, and
$*_i$ to label the branches of the non-deterministic operator $\circledast$.
Thus, the whole label set of the program, denoted by $\Sigma$, is given by standard labels of imperative programs, non-deterministic labels $*_i$ and probabilistic labels $\lprobbranch$ and $\rprobbranch$.

\begin{definition}[Probabilistic Control-Flow Automaton]
	A probabilistic control-flow automaton\footnote{see the appendix for discussion on models} is a tuple $(L, \Sigma, \delta,
    \ell_0, \ell_e)$, where $L$ is a \emph{finite} set of locations (also
    called nodes); $\Sigma$ is a \emph{finite} set of labels; $\delta \subseteq L \times \Sigma \times L$ is the transition relation, specially we may write $\ell \xrightarrow{\sigma} \ell'$ to denote $(\ell, \sigma, \ell') \in \delta$; $\ell_0, \ell_e \in L$ are the initial and ending location respectively.
\end{definition} 

We call a word of labels $\Sigma^*$ to be a \emph{trace}, and the accepting
words of a CFA \emph{accepting trace}s of that CFA. 
%\mytodo{define it, if
%there are spaces}{
The semantics, or the interpretation, of a trace
%} 
is a function $\interpret{-}$ that maps a trace to a function of program states. 
For a given finite trace $\tau$ of a CFA whose label set is denoted by $\Sigma$, its interpretation $\interpret{\tau}$ is inductively defined by: 
$$ \interpret{\sigma :: \tau}(s) := \interpret{\tau}(\interpret{\sigma}(s)) \qquad \interpret{\varepsilon}(s) := s $$
for each $\sigma \in \Sigma$, we adopt the standard interpretations for non-random labels and 
let nondeterministic and probabilistic labels have the same semantics as $\stskip$.

The \emph{weight}~
\footnote{probability of a trace should also take into account the semantics. Hence we use ``weights'', a purely syntactic notion.} 
of a trace $\tau$ denoted by $wt(\tau)$ is $2^{-n}$,
given that each probabilistic label means a fair binary probabilistic choice, 
and $n$ is the number of probabilistic labels within $\tau$.

The conversion from program to PCFA is standard for the non-random part. For probabilistic choices, we \emph{split} the branches into two, tagged with $\texttt{L}$ and $\texttt{R}$ respectively, and tag each statement $s_1 \oplus s_2$ with a unique identifier. Likewise, we split and tag non-deterministic branches with unique identifiers for each single branch. 
%\mytodo{???}{
An example of such translation can be seen in \cref{figure: motivation example cfa} above.
%}

In general, CFA definition does not require determinism. However, it is the case for the converted ones, as suggested below:

\begin{fact}
	\label{theorem: P is DFA}
	For any program $P$, the resulting CFA from conversion above is a deterministic finite-state automata (DFA).
\end{fact}

\paragraph{CFMDP, CFMC and Strategy}
For readers familiar with standard probabilistic models, it seems trivial to observe that a CFA like the one in \cref{figure: motivation example cfa} can be easily seen as a \emph{Markov decision process}. In such models, one could view edges with labels except for probabilistic ones as actions attached with a \emph{Dirac} distribution and probabilistic choices as actions named by the identifier with a fair Bernoulli distribution.
To make full use of this observation, we present the concepts of \emph{control-flow Markov chain (CFMC)} and \emph{control-flow Markov decision process (CFMDP)} and also \emph{strategy}, \emph{action} thereof. 
These notions hugely benefit exposition, which allows us to discuss problems within a single entity where techniques of both sides are easily applied.
Further, they serve as a proper extension to standard probabilistic models, which furnishes standard models with control flow as semantics.

A CFMDP is simply a \emph{deterministic} PCFA with no edges out of the ending location.
Note that we relax the restriction that $\lprobbranch$ and $\rprobbranch$ for a specific $i$ must appear in pairs.
Intuitively, a node with only one branch means getting to a dead node that could never reach the accepting node.
Furthermore, any CFA converted from a program is also a CFMDP.
We call identifiers for probabilistic distributions and non-probabilistic labels \emph{actions} of CFMDP.
Then, a CFMC is just a CFMDP that there is at most one action associated with every single node.
Next, the concept of \emph{strategy} for CFMDP mimics exactly the standard one. 
Specially, in this work, we only refer to the \emph{finite memory} ones.
In the following, when we say strategy, we refer to \emph{finite memory strategy}.

We use $\scriptS(\scriptA)$ to denote the set of all finite memory strategies of a CFMDP $\scriptA$.
Furthermore, we denote the CFMC induced by applying a strategy $\psi$ to a CFMDP $\scriptA$ as $\scriptA^\psi$.
 
 Intuitively, a finite memory strategy displays program states behind control flows.
 As to be mentioned below, we care only about terminating traces, so infinite memory strategies are pointless.

\paragraph{Problem}
In this work, the central problem we aim to solve is that: given a program $P$, a specification $(\varphi_e, \varphi_f)$ and $\beta$, whether the probability of $P$ starting at program states satisfying $\varphi_e$ end up with a program state \emph{not} satisfying $\varphi_f$, \emph{conditioning on the termination} of $P$ is $\le \beta$.

Let the indicator function be written as $[-]$,
formally, we define the probability of violation for a program with respect to
a specification $(\varphi_e, \varphi_f)$, and more generally for a CFMDP $\scriptA$ to be:
\begin{align}
	\prob[(\varphi_e, \varphi_f)]{\scriptA} := 
	\max_{s \models \varphi_e} \sup_{\psi \in \scriptS(\scriptA)} \sum_{\tau \in \scriptL(\scriptA^\psi)} wt(\tau) \cdot [\interpret{\tau}(s) \models \neg \varphi_f]
	\label{eq: prob for CFMDP}
\end{align}

So, the problem, denoted by $\vdash_\beta \bs{\varphi_e}\ P\ \bs{\varphi_f}$
is then just to ask whether the inequation $ \prob[(\varphi_e, \varphi_f)]{P} \le \beta $ holds. \footnote{We do not explicitly distinguish the program and the translated CFA thereof in the work for the ease of exposition.}

\paragraph{Counterexample}
In the non-random case, a counterexample of a program $P$ against a
specification $(\varphi_e, \varphi_f)$ is just an executable trace $\tau$ of $P$ such that there exists an initial program state $s \models \varphi_e$ that $\interpret{\tau}(s) \models \neg \varphi_f$. 
It is, however, not the case in the probabilistic context as we may need several traces. As discussed in \cite{pcegar}, a counterexample in probabilistic context should be an actual \emph{probabilistic} execution of the program. It means that a counterexample should consist of a certificate predicate that implies the pre-condition and a possible execution strategy of this program. 
In our context, we take the \emph{effect} of strategy application and formally present it to be: 
a counterexample for specification $(\varphi_e, \varphi_f)$ with threshold $\beta$ against a program $P$, is defined to be a pair $(S, \varphi_{err})$ where $S \subseteq \scriptL(P)$ is a set of traces within which any two traces have the same prefix until a probabilistic label and $\varphi_{err}$ is a proposition called \emph{error pre-condition} that is not effectively $\texttt{False}$, with $\varphi_{err} \to \varphi_e$ and $ \forall s \models \varphi_{err}.\forall \tau \in S.\interpret{\tau}(s) \models \neg \varphi_f $.

\paragraph{Notations}
For the convenience of exposition, throughout the paper, unless specially stated, we fix a program $P$ and a specification $\Phi := (\varphi_e, \varphi_f)$. If it is clear from the context, we write just $\prob{-}$ to denote $\prob[\Phi]{-}$.

\section{\textsc{ProbTA}: the Proof Rule}
\label{sec: the proof rule}

% !TEX root = ./main.tex

In this section, we first present an abstract rule from a pure probabilistic and trace-based perspective,
concretise it and illustrate the major obstacle towards actual application in \cref{subsec: pre-proof rules}.
We then give a proper solution to this obstacle in \cref{subsec: probabilities from trace view}.
Finally, we present the applicable proof rule \textsc{ProbTA} and show its soundness \& completeness in \cref{subsec: the proof rule}.

\subsection{General Rules}
\label{subsec: pre-proof rules}
\paragraph{The Abstract Proof Rule}
The abstract proof rule for combining probabilities and trace abstraction is a \emph{tautology} over \emph{any} possible probabilistic measures on a given measurable space over a label set $\Sigma$.
It serves as a probabilistic theoretic foundation of rules below.

More specifically, let $(\Theta_\Sigma, \scriptF)$ be a valid measurable space on traces, 
where $\Theta_\Sigma := \Sigma^*$ is the set of all finite traces over $\Sigma$, 
for any probability measure $\mu_\Phi$ defined to be called violation probabilities of a set of traces,
we have: 
for any two $\scriptF$-measurable sets $\Theta$ and $\Theta'$,
if $\Theta \subseteq \Theta'$, then
$\mu_\Phi(\Theta) \le \mu_\Phi(\Theta')$.

We then denote $\mu_{(\varphi_e, \varphi_f)}(\Theta) \le \beta$ to be:
$$ \vdash_\beta \bs{\varphi_e}\ \Theta\ \bs{\varphi_f} $$
So, this mathematically trivial fact then leads to an abstract proof rule that:
\begin{equation}
	\frac{\Theta \subseteq \Theta' \qquad \vdash_\beta \bs{\varphi_e}\ \Theta'\ \bs{\varphi_f}}{\vdash_\beta \bs{\varphi_e}\ \Theta\ \bs{\varphi_f}}
	\tag{\textsc{AbsProbTA}}
	\label{eq: the abstract proof rule}
\end{equation}

Towards the application of this abstract rule, there is one more obstacle we need to overcome -- 
a proper definition of $\mu_\Phi$ to bridge between the standard probability definition and that of a set of traces.
This can be seen more clearly by a concretisation to the abstract rule.

\paragraph{An Applied Proof Rule}
Let the $\sigma$-algebra $\scriptF$ over $\Theta_\Sigma$ be \emph{regular languages}, as the concrete hosts to discuss in our case is PCFA.
The validity comes from standard facts on regular languages.
Then, as described in \cref{sec: motivation example and overview}, we create two automata, 
one keeps non-violating trace, called the \emph{certified module}, denoted by $Q$, 
one keeps an over-approximation of violating traces called the \emph{violating module}, denoted by $A$.
So we have the following applied proof rule for a program $P$:
\begin{equation}
	\frac{\begin{matrix}
	\scriptL(P) \subseteq \scriptL(Q) \cup \scriptL(A) \\ \vdash \bs{\varphi_e}\ Q\ \bs{\varphi_f} \quad \vdash_\beta \bs{\varphi_e}\ \scriptL(A)\ \bs{\varphi_f}
	\end{matrix}}{\vdash_\beta \bs{\varphi_e}\ \scriptL(P)\ \bs{\varphi_f}}	\tag{\textsc{PreProbTA}}
	\label{eq: the pre-proof rule}
\end{equation}
where $\vdash \bs{\varphi_e}\ Q\ \bs{\varphi_f}$ is
an abbreviation of 
that:
for all $\tau \in \scriptL(Q)$, the Hoare triple $\bs{\varphi_e}\ \tau\ \bs{\varphi_f}$ holds.

To make a truly applicable proof rule, there is a final obstacle to address:
from $\vdash_\beta \bs{\varphi_e}\ \scriptL(P)\ \bs{\varphi_f}$ to $\vdash_\beta \bs{\varphi_e}\ P\ \bs{\varphi_f}$.
This requires, as mentioned above, a proper definition of $\mu_\Phi$.

\subsection{Probabilities from Trace View}
\label{subsec: probabilities from trace view}

% !TEX root = ./main.tex

For such a supportive definition, it suffices to aim simply for enlargement -- $\mu_\Phi(\scriptL(P)) \ge \prob[\Phi]{P}$.
We, however, aim for more -- strict equality, which guarantees the completeness of our proof rule.

Our solution for this definition comes from mimicking how probability is computed in \cref{eq: prob for CFMDP}. 
In the formula, one computes the \emph{least upper bound} of violation probabilities of \emph{all possible CFMC} within the CFMDP. 
Following this observation, the basic idea behind our solution is that: 
we try \emph{merging} traces into CFMC and use the least upper bound of violation probabilities among all these CFMC. 
To write it more concretely, rather than the general probability measure $\mu_\Phi$, 
this probability definition is denoted below by $\prob[\Phi][\scriptL]{-}$ instead.

Additionally, to confirm the meaningfulness of this definition, one may need to observe the trivial fact:

\begin{fact}
	\label{theorem: prob of CFMC independent of shape}
	For any two CFMCs $\scriptM$ and $\scriptM'$, if $\scriptL(\scriptM) = \scriptL(\scriptM')$, then $\prob[\Phi]{\scriptM} = \prob[\Phi]{\scriptM'}$.
\end{fact}

Based on this idea, we first need to define the notion of \emph{mergeable}.
A set of paths $\Theta$ is called \emph{mergeable} iff there exists a CFMC $\scriptM$ such that $\Theta = \scriptL(\scriptM)$.
We call all mergeable subsets of a given path set $\Theta$ to be $\mergeable(\Theta)$.
	
So, given a set of paths $\Theta$ and specification $(\varphi_e, \varphi_f)$, we define the violation probability of this set to be:
\begin{align}
	\label{eq: prob of trace set}
	\prob[(\varphi_e, \varphi_f)][\scriptL]{\Theta} := \max_{s \models \varphi_e} \sup_{\Pi \in \mergeable(\Theta)} \sum_{\tau \in \Pi} wt(\tau) \cdot [\interpret{\tau}(s) \models \neg \varphi_f]
\end{align}
	
So, by definition, we have for any CFMDP $\scriptA$:
\begin{align}
	\label{eq: true def <= trace def}
	\prob[\Phi]{\scriptA} \le \prob[\Phi][\scriptL]{\scriptL(\scriptA)}
\end{align}

To reach the equality, we naturally need another inequation from another direction.
%However, it is not generally the case that any CFMDP of any structure will satisfy the inequation for the opposite side, as shown in \cref{sec: proof of isomorphism}. 
To show this, we will technically require that for any node of the given CFMDP, if $\probbranch_{(i, \lbranch)}$
and $\probbranch_{(i, \rbranch)}$ exist at the same time for some $i$, then they must point to different locations.
In effect, this requirement is not an actual restriction.
As for any CFMDP, we can perform \emph{normalisation} to reach this effect while keeping the accepting traces unchanged.
Concrete method for normalisation is given in the Appendix.

As this restriction will not affect expressivity, to ease exposition, we will assume all CFMDP to discuss below to be normalised and satisfy the restriction.
With this restriction, we have the following theorem:
	
\begin{theorem}
	\label{theorem: subset CFMC and strategy isomorphism}
	For any (normalised) CFMDP $\scriptA$, and any CFMC $\scriptM$, if $\scriptL(\scriptM) \subseteq \scriptL(\scriptA)$, then there exists a strategy $\psi$ such that $\scriptL(\scriptA^\psi) = \scriptL(\scriptM)$.
\end{theorem}

The key to a constructive proof relies on the normalisation above.
The idea is that we use a series of \emph{dual} states in the strategy to delay making actual decisions, and
final decision is \emph{effectively} made based also on the actual next location of $\scriptA$ chosen.
A full proof is attached in the Appendix.

With this theorem, we then obtain:
\begin{align}
	\label{eq: true def >= trace def}
	\prob[(\varphi_e, \varphi_f)]{\scriptA} \ge \prob[(\varphi_e, \varphi_f)][\scriptL]{\scriptL(\scriptA)}
\end{align}
	
By \cref{eq: true def <= trace def,eq: true def >= trace def}, we reach the desirable effect:
	
\begin{theorem}
	\label{theorem: coincidence of two probabilities for CFMDP}
	For any (normalised) CFMDP $\scriptA$, its violation probability is the same as the violation probability of its traces. That is:
	\begin{align}
		\prob[(\varphi_e, \varphi_f)]{\scriptA} = \prob[(\varphi_e, \varphi_f)][\scriptL]{\scriptL(\scriptA)}
	\end{align}
\end{theorem}

In the following,
if it's clear from context, we ignore the subscript and superscript of the probabilistic operator -- that we abbreviate $\prob[\varphi]{-}$ or $\prob[\varphi][\scriptL]{-}$ to be just $\prob{-}$.

\begin{remark}
	This desirable result of equality owes partially to our definition of PCFA.
	In this definition, traces keep partial information on the original structure.
	This arrangement helps us re-construct a canonical form of models.
\end{remark}

\subsection{\textsc{ProbTA}: the Proof Rule}
\label{subsec: the proof rule}

Now, we are ready to present our proof rule and its properties with proof formally.

The proof rule is presented by that:
given a program (or, equivalently, its PCFA) $P$, a PCFA $Q$ called the certified module, and another PCFA $A$ called the violating module, we have:
%\todo{Formalisms of $P, Q$ and $A$}

\begin{equation}
	\frac{\begin{matrix}
	\scriptL(P) \subseteq \scriptL(Q) \cup \scriptL(A) \\ \vdash \bs{\varphi_e}\ Q\ \bs{\varphi_f} \quad \vdash_\beta \bs{\varphi_e}\ \scriptL(A)\ \bs{\varphi_f}
	\end{matrix}}{\vdash_\beta \bs{\varphi_e}\ P\ \bs{\varphi_f}}	\tag{\textsc{ProbTA}}
	\label{eq: the proof rule}
\end{equation}

Intuitively, this is to say, if one wants to prove the Hoare-triple $\bs{\varphi_e}\ P\ \bs{\varphi_f}$ holds with violation tolerant threshold $\beta$, 
it suffices to find two automatons $Q$ and $A$ such that 
traces in certified module $Q$ all satisfy the triple, and the set of accepting traces of violating module $A$ satisfies the triple up to violation probability $\le \beta$. 
Furthermore, think of the two modules $Q$ and $A$ as proper programs. This implies that one can show a probabilistic program's satisfiability via some ``near" programs that over-approximate the original program in execution traces.

The proof rule is sound \& complete where the soundness is given by:
\begin{theorem}
	\label{theorem: soundness of the proof rule}
	For any program $P$ and given specification $(\varphi_e, \varphi_f)$ with threshold $\beta$, if there exists such control-flow automatons $Q$ and $A$ that: $\scriptL(P) \subseteq \scriptL(Q) \cup \scriptL(A)$ with $\vdash \bs{\varphi_e}\ Q\ \bs{\varphi_f}$ and that $\vdash_\beta \bs{\varphi_e}\ \scriptL(A)\ \bs{\varphi_f}$, then we have $\vdash_\beta \bs{\varphi_e}\ P\ \bs{\varphi_f}$.
\end{theorem}
\begin{proof}
	The proof is given by the following (in)equation chain:
	\begin{align*}
		& \prob{P} &
		\\
		= \ & \prob{\scriptL(P)} & (\text{\cref{theorem: coincidence of two probabilities for CFMDP}})
		\\
		\le \ & \prob{\scriptL(Q) \cup \scriptL(A)} & 
		\\
		\le \ & \prob{\scriptL(Q)} + \prob{\scriptL(A)} & (\text{Probabilistic Union Bound})
		\\
		= \ & \prob{\scriptL(A)} \le \beta
	\end{align*}
\end{proof}

The proof for completeness credits mainly to \cref{theorem: coincidence of two probabilities for CFMDP}, which is given by:
\begin{theorem}
	\label{theorem: completeness of the proof rule}
	For any program $P$ and given specification $(\varphi_e, \varphi_f)$ with threshold $\beta$, if $\vdash_\beta \bs{\varphi_e}\ P\ \bs{\varphi_f}$ then there exists such control-flow automatons $Q$ and $A$ that: $\scriptL(P) \subseteq \scriptL(Q) \cup \scriptL(A)$ with $\vdash \bs{\varphi_e}\ Q\ \bs{\varphi_f}$ and that $\vdash_\beta \bs{\varphi_e}\ \scriptL(A)\ \bs{\varphi_f}$.
\end{theorem}
\begin{proof}
	We just let $Q$ to be an empty automaton and let $A$ be the control-flow automaton from $P$, so we have $\prob{\scriptL(A)} = \prob{\scriptL(P)} = \prob{P} \le \beta$.
\end{proof}

\section{An Automated Algorithm via CEGAR}
\label{sec: algorithm}

% !TEX root = ./main.tex

With insight from the soundness of \ref{eq: the proof rule}, we present a new automated probabilistic program verification method driven by the famous \emph{CounterExample-Guided Abstraction Refinement (CEGAR)} procedure, which can be seen as an extension to the CEGAR-driven non-random trace abstraction algorithm \cite{trace_abstraction}.

In this section, we first present the algorithm's main process and then present its details. Finally, we present some discussion on the ability and limitation of this algorithm and discussion on comparison with \cite{trace_abstraction_modulo_probability}, which is another extension to \cite{trace_abstraction}.

\subsection{Framework}
The framework of our algorithm is presented in \cref{figure: algorithm
framework}. 
It consists mainly of two loops. 
In the right loop, normal non-random trace abstraction is performed.
This process updates
%\mytodo{call it
%satisfying automaton, or other name}{the automaton $Q$} 
the certified module $Q$ with more non-violating traces.
The modification to $Q$ must maintain $\vdash \bs{\varphi_e}\ Q\ \bs{\varphi_f}$.
The left loop manipulates the violating module $A$. During the process, $A$ will typically undergo expansions and shrinks to classify more traces and find a possible counterexample. 
Additionally, $A$ that takes part in the emptiness check $\scriptL(P \setminus Q \setminus A) = \emptyset$ in \cref{figure: algorithm framework} must always have that $\vdash_\beta \bs{\varphi_e}\ A\ \bs{\varphi_f}$.

\begin{figure}
	\includegraphics[width=\linewidth]{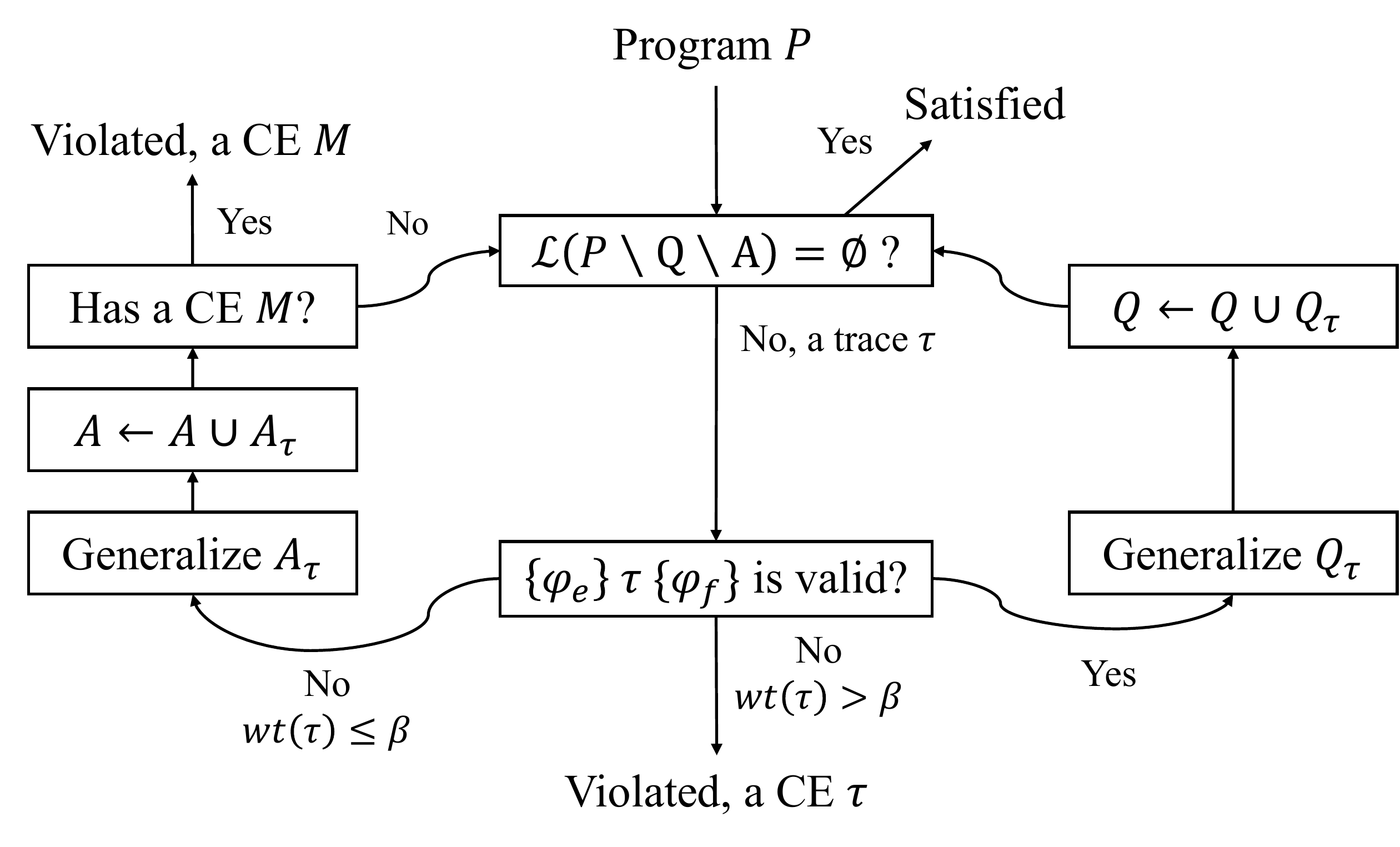}
	\caption{Probabilistic Trace Abstraction}
	\label{figure: algorithm framework}
\end{figure}

With these principles in panorama, let's get one step closer.
To begin with, we check whether traces of $P$ is covered by $Q$ and $A$,
if it is, by the proof rule and principles above, we return SAT as a result; 
if not, we further proceed with our process by picking a trace with \emph{minimal} length and try classifying the trace.

The classification is based on:
1) whether it is possible (exists a state $s \models \varphi_e$) to reach a wrong state that satisfies $\neg \varphi_f$ and
2) the weights of the trace.
The trivial case is when the found trace $\tau$ is violating, with weight $> \beta$, this means we have found a valid counterexample, 
so we just return UnSAT along with a counterexample consists of this trace and the error pre-condition $\varphi_e \wedge wp(\tau, \neg \varphi_f)$.
Here, $wp$ is the famous \emph{weakest pre-condition} \cite{dijkstra1975guarded} operator.
By conjunction of $\varphi_e$ and $wp(\tau, \neg \varphi_f)$, we identify the exact condition a program state needs to satisfy to start from this trace to reach a wrong state.

Then, if the trace $\tau$ is non-violating, we generalise it into a CFA $Q_\tau$ and enrich the certified module $Q$ with $Q_\tau$.
The traces accepted by $Q_\tau$ are all non-violating.
Intuitively they are considered non-violating with ``near reasons'' to $\tau$.
In other words, we somehow ``learn'' the ``reason'' of non-violation of the found trace $\tau$ and try to \emph{reuse} this reason to a set as much as possible. 
Details of this generalisation are presented in \cref{subsec: non-probable generalisation}.
	
The most interesting case is when the trace $\tau$ is violating, but with weight $\le \beta$.
In this case,
another kind of generalisation is performed on $\tau$ to form a CFA $A_\tau$. 
Effectively, this generalisation differs from the above one mainly in that rather than a full proof, 
we here only learn a ``semi-'' one.
This means we allow non-violating traces to enter $A_\tau$ rather than pure violating ones.
Intuitively, this generalisation just finds traces of ``near reason'' to $\tau$ but does not strictly verify their violation.
Details on the generalisation of this kind are introduced below in \cref{subsec: probable generalisation}.
After generalisation, $A$ is expanded with traces in $A_\tau$.
The expansion is followed by an examination on $A$ aiming at either finding a possible counterexample or confirming non-violating.
During the process, \emph{shrinks} may be performed on $A$, in contrast to the non-violating case where only expansions are performed.
Details on this process are presented in \cref{subsec: examine A}.

\subsection{Generalisation for Non-Violating Traces}
\label{subsec: non-probable generalisation}
In this part, we briefly describe generalisation in the right loop where a certified automaton $Q_\tau$ is constructed from non-violating $\tau$.
In this construction, we require $Q_\tau$ to contain only non-violating traces and $\tau$ is in $Q_\tau$.
This is essentially the same as the non-random trace abstraction \cite{trace_abstraction}.

\paragraph{Floyd-Hoare Automaton}
The host of our generalisation is, as in the non-random case, a Floyd-Hoare automaton, where every location is attached with a proposition, and labels between locations should satisfy the Hoare triple between.
Formally, a Floyd-Hoare automaton is a pair $(\scriptA, \lambda)$ where $\scriptA$ is a PCFA and $\lambda$ is a function from location set of $\scriptA$ to the set of propositions. And we require that for every $(\ell, \sigma, \ell')$ in the transition set of $\scriptA$, the Hoare triple $\bs{\lambda(\ell)}\ \sigma\ \bs{\lambda(\ell')}$ holds.

\paragraph{Generalisation}
The process of generalisation for a trace $\tau$ is mainly divided into three phases.
In the first step, one performs \emph{propositions insertion} on $\tau$ which forms a Floyd-Hoare automaton that has $\tau$ as the only trace.
We call this process ``propositions insertion'' as it is essentially about inserting (attaching) a proper proposition (to the location) between every two labels, the head and the end of $\tau$.
In the second step, one merge locations that have the same proposition.
In the last step, edges are added to the Floyd-Hoare automaton formed in the previous step, while maintaining the validity of the Floyd-Hoare automaton -- 
an edge with label $\sigma \in \Sigma$ is added from $\ell$ to $\ell'$ iff $\bs{\lambda(\ell)}\ \sigma\ \bs{\lambda(\ell')}$ is valid, 
where $\Sigma$ is the label set of current interests, and $\lambda$ is the proposition assigning function of the Floyd-Hoare automaton.

\paragraph{Propositions Insertion for Non-Violating Traces}
As mentioned above, proposition inserting is the first process in generalisation. 
Intuitively, this means to ``find a reason'' for why the given trace is non-violating.
Flavour thereof can be found in \cref{figure: motivation example category 1,figure: motivation example category 2}.
In non-randomised trace abstraction, this is usually done by \emph{interpolation}.
A bunch of possible interpolation techniques are available, e.g.~Craig interpolation \cite{craig_interpolation}, a new technique by Christ et al.~\cite{smt_interpol}.

\subsection{Generalisation for Violating Traces}
\label{subsec: probable generalisation}

The generalisation for violating traces is mainly the same as described above for non-violating traces except for the following differences:

\paragraph{Propositions Insertion}
To begin with, the usual technique for non-violating traces -- interpolation, is not directly applicable for violating traces, for it is a procedure for \emph{UnSAT} propositions. 
We here handle the insertion with \emph{weakest pre-condition} computation -- we compute backwardly from $\neg \varphi_f$ and insert the weakest pre-condition as the proposition before a label.
At the head of the trace, we use $\varphi_e \wedge wp(\tau, \neg \varphi_f)$ as the proposition with the same reason introduced above.

\paragraph{Generalised Result For Violating Traces}
As another major difference to mention, which is also a deliberate result of our propositions insertion method, 
the generalised result for violating traces is radically different from the non-violating one --
the resulting Floyd-Hoare automaton cannot guarantee general violation of traces thereof.
As a technical detail, this is mainly because traces with self-contradictory semantics may arise, 
for example, a trace like $...$ \cbx{$ X ::= 0$} \cbx{$ \assume\ X > 0$} $ ...$ may appear if the propositions around do not mention the variable $X$.
However, we in general deliberately allow such traces. 
This is because: 
in general, violating traces of a program is not a regular language and hence not expressible exactly with PCFA, 
so we must allow to some extend over-approximations. 
However, this again brings about some new problems, details of which are discussed below in \cref{subsec: property of the algorithm}.
In principle, this design is essentially a result of balancing.

\subsection{Examine $A$}
\label{subsec: examine A}
The final part is about how to examine the obtained $A$ to either prove its non-violation or extract a counterexample therefrom. 
Before this process, we assume that $A$ contains only traces in $P$, 
as it is pointless to discuss with traces not in $P$. 
Technically, this is done by intersecting $A$ with $P$, minimising the result followed by normalisation. 
Then, we perform this examination via iteration whose framework is presented in \cref{figure: examine A}.

\begin{figure}
	\includegraphics[width=\linewidth]{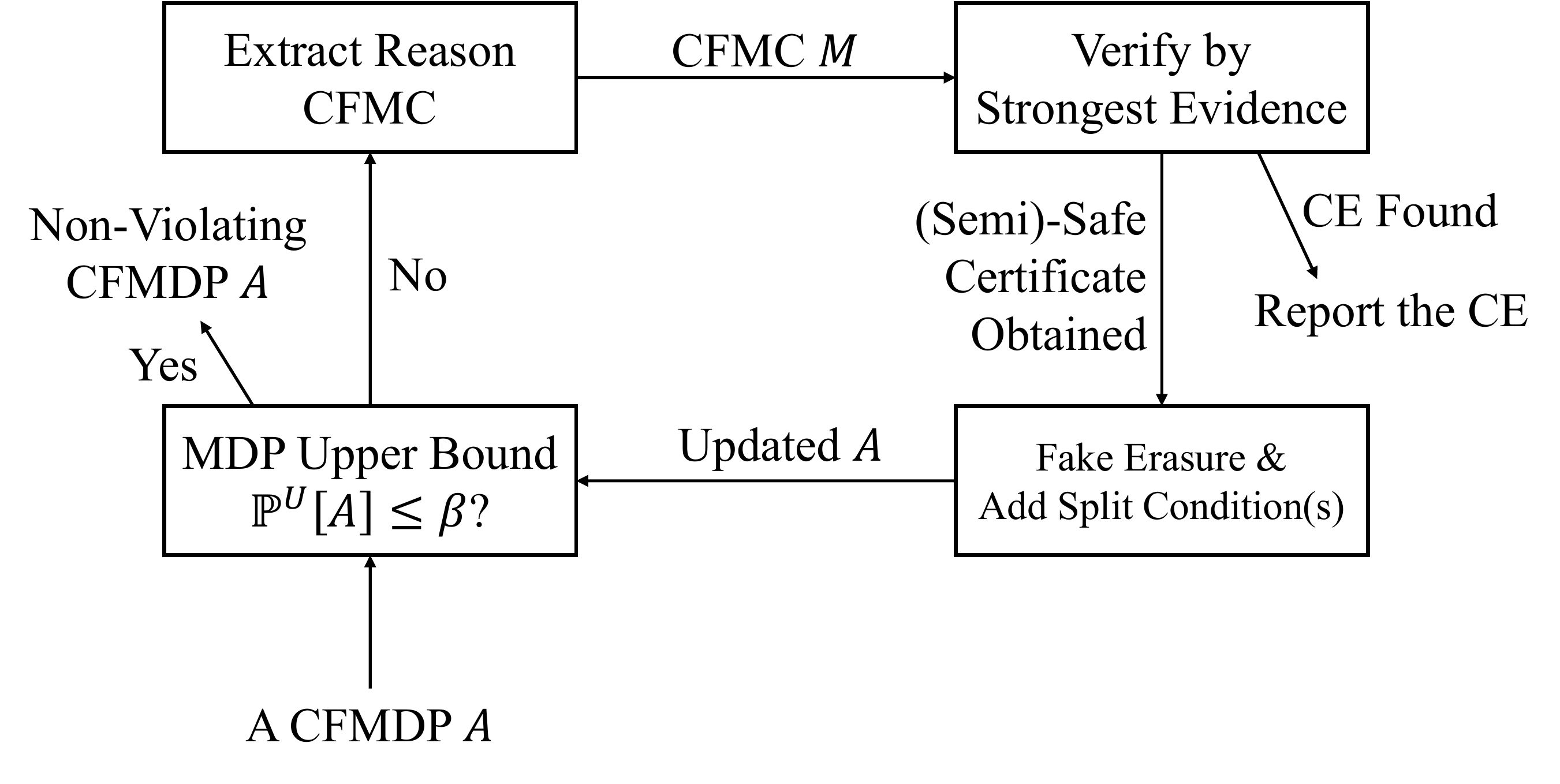}
	\caption{Examine A}
	\label{figure: examine A}
\end{figure}

\paragraph{MDP Upper Bound of CFMDP \& Reason CFMC}
The first obstacle to overcome is to prove non-violation of a given CFMDP $\scriptA$.
A natural idea is that we compute an upper bound of the maximum violating probability, and show that this upper bound $\le$ threshold $\beta$. 
For this upper bound, the definition of definition brings about a possible solution: 
observe that a proper upper bound can be that we just set the value of the indicator function $[-]$ to be $1$ in \cref{eq: prob for CFMDP} -- 
that is: we define the \emph{MDP upper bound} $\prob[][U]{\scriptA}$ of a CFMDP $\scriptA$ to be:
\begin{align}
	\prob[][U]{\scriptA} := &\ \max_{s \models \varphi_e} \sup_{\psi \in \scriptS(\scriptA)} \sum_{\tau \in \scriptL(\scriptA^\psi)} wt(\tau) \cdot 1 \\
	= &\ \sup_{\psi \in \scriptS(\scriptA)} \sum_{\tau \in \scriptL(\scriptA^\psi)} wt(\tau)
	\label{eq: prob of upper bound of CFMDP}
\end{align}

Notice that the \cref{eq: prob of upper bound of CFMDP} is just the \emph{maximum reachability probability} of an MDP, 
which can be computed by existing standard methods. 

Adopting this method, we then call the underlying CFMC (which comes from a specific strategy) that produces this maximum probability to be the \emph{reason} CFMC.
Next, we examine this CFMC to see if it is a true violating structure, as this CFMC is just an over-approximation.

\paragraph{Verify by Strongest Evidence}
The verification of the semantics of such a complex structure like CFMC that may contain loops is hard.
Such hardness is discussed more deeply in \cite{pcegar}.
In this work the authors declared that such complex structures are not ``directly amenable to conventional methods''.
Confronted by a similar difficulty like them, we adopt in this part their general idea of verifying via \emph{strongest evidence}.
%Confronted with this difficulty, this work adopts the idea of \emph{verifying via strongest evidence}.
%Under these circumstances, in this part, we 
%In this part, we further follow the general idea of using \emph{strongest evidence} by \cite{pcegar}.
%We further adopt the same method as theirs -- we verify the given CFMC using  \cite{trace_enumeration}. 
More concretely, we verify traces one by one according to their weights, 
with the algorithm presented in \cite{trace_enumeration} like \cite{pcegar}, 
and see whether we can find a counterexample or not.
Notwithstanding, for that our model is different from theirs,
details we are going to tackle is also different.

Before delving into details of this method, we first see what \emph{kinds} of traces may exist during the verification. 
Rather than simply the verifiably violating and non-violating traces, one problem is \emph{compatibility} between traces -- 
the path conditions between different violating traces may not be compatible! For example, two violating traces for \cref{example: motivation}, namely 
\cbx{$X ::= 0$} \cbx{$ \rprobbranch[0]$} \cbx{$ \stskip$} \cbx{$ \assume\ C > 0$} \cbx{$ \rprobbranch[1]$} \cbx{$ X ::= X + 1$} \cbx{$ \assume\ \neg (C > 0)$}
and 
\cbx{$X ::= 0$} \cbx{$ \rprobbranch[0]$} \cbx{$ \stskip$} \cbx{$ \assume\ C > 0$} \cbx{$ \rprobbranch[1]$} \cbx{$ X ::= X + 1$} \cbx{$ \assume\ C > 0$} \cbx{$ \rprobbranch[1]$} \cbx{$ X ::= X + 1$} \cbx{$ \assume\ \neg (C > 0)$}
are not compatible, as the first trace requires $C$ to be exactly $1$ while the second one requires $C$ to be exactly $2$.

\medskip

However, in this way, there are actually $2^n + 1$ kinds of traces during the verification, 
where $n$ is the number of violating traces, and the final $1$ is for non-violating traces, 
which is called \emph{fake} in the following. 
To avoid the possible explosion from this, we optimise this process, and let there be just $3$ kinds -- 
we set a \emph{canonical} set of traces called the \emph{mainstream} set, which is updated on-the-fly.
So a trace during the process is either 1) a violating \& compatible trace, 2) a violating but incompatible trace or 3) a fake trace.
A trace is added to mainstream only when its path condition is compatible with \emph{all} those of traces currently in the mainstream. 
We call the conjunction of $\varphi_e$ and all path conditions inside mainstream to be \emph{total pre-condition}.

The termination condition of verifications can be briefly summarised as the following two cases:
The first is when the mainstream is enlarged to have total probability $> \beta$, which means a proper counterexample is found -- 
it consists of the mainstream set along with the error pre-condition total pre-condition.
the second is when the accumulated probabilities of fake or incompatible traces get $\ge p - \beta$, 
where $p$ is the current upper bound value of violation
This means even if all left traces are violating, its probability will not be $> \beta$. 
And such accumulated set of fake or incompatible traces with their total probability are called \emph{(semi)-safe certificate}.
Details of the process including proof for termination is given in \cite{pcegar}.

One may then argue that the thus obtained safe certificate is not sound -- 
there may still exist underlying counterexample inside this CFMC found, 
but just with a different error pre-condition other than the total pre-condition of the current mainstream.
Yes, indeed, 
and that is why it is just called a \emph{semi}-safe certificate.
Furthermore, in order to make this optimisation not affect the total correctness, 
we'll need to \emph{add split conditions}, which is discussed right below, 
so to ensure a different form of the next possible reason CFMC, 
which allows us to find for a different mainstream.

\paragraph{Fake Erasure \& Split Conditions Addition}
If a (semi)-safe certificate is obtained, we then need to prevent the re-appearance of the same reason MC again. 
This involves two processes -- one for fake and one for violating but incompatible traces. 
The process for fake traces is trivial -- we again generalise it like in the right loop for $P$ and then erase them from the current $A$.

The process for incompatible traces is more interesting. 
In this case, what we essentially need is to prevent those incompatible traces from mixing into a single CFMC again. 
To reach this, the simplest thing to do, from the angle of how CFMC comes from CFMDP, is to let the two trace sets belong to different strategies, or even within the same strategy, a trace on another side should be a fake.
This is the idea behind our solution -- adding \emph{split conditions}. 
We add a new starting node before the initial location of $A$ and add two edges -- $\assume\ H$ and $\assume\ \neg H$ pointing to the original initial location of $A$ 
and set this new node as the starting node of $A$, where $H$ is the total pre-condition of the current mainstream.
After that, we erase the found mainstream and incompatible traces attached with $\assume\ \neg H$ and $\assume\ H$ respectively.
Furthermore, to avoid explosion from split conditions, we perform the optimisation that after adding the first split conditions, we modify the common initial assumption statement and cut it into two.

\subsection{Properties of the Algorithm}
\label{subsec: property of the algorithm}
After the above description to details of the algorithm framework, we then discuss the properties of our algorithm.

\paragraph{Soundness \& Completeness}
By the proof rule \cref{eq: the proof rule}, the algorithm framework, given the ability of the underlying theory solver is guaranteed, if an answer is produced, the answer is sound and complete.

\paragraph{Termination \& Refutational Completeness}
As may have been noticed, the algorithm framework above just presents a \emph{semi}-algorithm that may never terminate. This phenomenon comes from that: 
1) in the right loop, the interpolation may not find a truly useful ``reason'' of satisfiability, and thus the enumeration on the RHS may not end;
2) in the left loop, although repeatedly picking the same reason CFMC is avoided, 
it is not guaranteed that the upper bound of $A$ will generally decrease after finite iterations of the process described in \cref{subsec: examine A}.

Theoretically, we can modify the above algorithm to a more ``complete'' one. 
Despite the theoretical merit, this modification seems to compromise both the ability to show satisfiability and the overall efficiency.
We ignore generalisation of $A$ and rather than requiring $A$ and $Q$ to fully cover $P$, we just let $A$ be a repository that keeps violating traces that have already been found,
and when the upper bound of $P \cup A$ is $\le \beta$, the algorithm produces SAT. 
This modification essentially guarantees the termination of the process for $A$ -- as this is just a process dealing with \emph{finite} traces.
In this way we have the good theoretical property of \emph{refutational completeness} in termination:
\begin{theorem}
	\label{theorem: modified algorithm refutationally complete}
	For any program $P$, and specification $(\varphi_e, \varphi_f)$ with threshold $\beta$, if $P$ does not satisfy $\vdash_\beta \bs{\varphi_e}\ P\ \bs{\varphi_f}$, then the above modified algorithm will eventually terminate and produce a counterexample.
\end{theorem}
\begin{proof}
	If $\not\vdash_\beta \bs{\varphi_e}\ P\ \bs{\varphi_f}$, there must exist such a CFMC that (and whose trace set) has violating probability $> \beta$.
	By the theorem of the existence of the strongest evidence proved in \cite{trace_enumeration}, there must exist such a finite set that has total probability $> \beta$.
	Then, by that in any $P$, the traces of a certain finite length is finite, and that after an iteration, a picked trace will be either in $A$ or $Q$ and not in the new $P$, 
	we have that after a certain finite steps of iterations, the algorithm either terminates or $P$ does not contain any trace of length less than or equal to a certain given length -- 
	recall that in the main loop, every time we pick a trace among the shortest length.
	By the facts above, the theorem is easily shown by taking the maximum length $N$ of the existent counterexample, and the algorithm must terminate when all traces in $P$ have length $> N$.
\end{proof}

\paragraph{Comparison with \cite{trace_abstraction_modulo_probability}}
Smith et al. proposed another algorithm that extends trace abstraction with probability based on program synthesis.

As is discussed in \cref{sec: intro}, one fundamental difference is the style of modeling randomness.
In general, the range of problems that can be solved by our algorithm and theirs is incomparable. 
Their method is capable of solving problems with parameterised probabilities and thresholds, while our current algorithm is only compatible to constant ones. 
Considering programs with constant threshold, our algorithm is complete for finite trace programs, with proof similar to \cref{theorem: modified algorithm refutationally complete}. However, it's not the case for their method.
The inaccuracy comes from the \emph{union bound} principle $\prob{A \cup B} \le \prob{A} + \prob{B}$ they use for accumulation of probabilities, which causes the violation probability to explode quickly, and even greater than $1$.
Probabilities are computed by accumulating feasible traces in our method, enabling more accurate computation for general cases where probabilities are co-related.
Furthermore, they use a new kind of automata called \emph{failure automata},
where probabilities of failure is also recorded on each locations. We adopt Floyd-Hoare automata utilised in non-probabilistic trace abstraction, and thus we are able to re-use most of the existing facilities.

\section{Related Work}
\label{sec: related work}

% !TEX root = ./main.tex

\subsection{Related Works}
\paragraph{Probabilistic Imperative Program Verification}
The work by Smith et al.~
\cite{trace_abstraction_modulo_probability}
is the most relevant work in our direction and
is the pioneering work towards combining probabilistic program verification and trace abstraction.
As a comparison on algorithm framework is given above in
\cref{subsec: property of the algorithm},
we here focus on a more theoretical side.
In \cite{trace_abstraction_modulo_probability}, the authors proposed a proof rule based on data randomness.
Their rule is also sound and complete but is relatively trivial to prove,
as in data randomness, definition of program violation is equivalent to a simple sum over all possible traces.
Despite mathematical trivialness of proof, automation is hard to obtain.
While in our case, although proof to validity of the rules is non-trivial,
it provides new insights between CFA and set of traces in the control-flow randomness style.
Also, the proof rule leads to a relatively natural way to work with the CEGAR automation framework.
In all, our work and theirs have different focuses and adopt different approaches.
A unified framework that combines advantages from both sides may therefore be a natural direction for future work.

On the other hand, Wang et al.~
\cite{exponential_analysis_of_probabilistic_program} analysed the exponential bounds, a sub-case of probabilistic imperative program verification.
They found a novel fixed point theorem to compute arguably accurate upper and lower bounds of assertion violation probability of a given probabilistic imperative program.

Hermanns et al.~\cite{pcegar} extended the famous \emph{counterexample-guided abstract refinement (CEGAR)} framework \cite{cegar} to the probabilistic context, focusing especially on \emph{predicate abstraction}. In our work, we make abstraction on the possible execution traces, while in \cite{pcegar}, the abstraction is made on program states.
Another connection between ours and theirs is, as stated above, we both adopt the idea of strongest evidence during analysis.
Cousot et al.~\cite{prob_abstract_interpreation} introduced the abstract interpretation framework to probabilistic context, while we introduced trace abstraction to that.
Wang et al.~\cite{pmaf} proposed the framework \textsc{pmaf}, an elegant algebraic framework which unifies both \emph{data randomness} and \emph{control-flow randomness} and
is inspiring in that it may provide hints on combining our algorithm and the one by Smith et al.~\cite{trace_abstraction_modulo_probability}.

\paragraph{Trace Abstraction \& Probabilistic Analysis}
Trace abstraction \cite{trace_abstraction} is a successful technique based on automata theory for non-probabilistic imperative programs verification.
Our work is an extension of this technique to the probabilistic case.
Besides the work by Smith et al.~\cite{trace_abstraction_modulo_probability}, another efforts to combine probabilistic analysis and trace abstraction is the recent work by Chen et al.~\cite{ast_by_omega_decompose}.
The authors analysed the \emph{almost sure termination (a.s.t.)} problem that studies whether a program terminates with probability $1$.

\paragraph{Probabilistic Model Checking}
The development of 
probabilistic model checking (PMC) has enjoyed a huge success in the past two decades. 
It's a cross-domain research field that model checking tools like PRISM \cite{prism} has been applied to multiple areas~\cite{prism_application_algorithm,prism_application_biology,prism_application_protocol1,prism_application_protocol2,prism_application_industry}.

In PMC, objects are mainly Markov chains or Markov decision processes.
In the last decade, new models with transcending expressivity like \emph{recursive Markov chain (RMC)}~\cite{rmc,rmc_init}, or equivalently, \emph{probabilistic pushdown automaton (pPDA)} \cite{ppda} and their model checking problems \cite{ltlmc_rmc,BEKK13} has attracted much attention.
Higher-order models \cite{phors_long,phors} are also analysed recently.

In our work, we newly proposed two models, namely CFMC and CFMDP, which can be seen as another kind of extension. As it is very natural and trivial to cast CFMDP to MDP (and so for CFMC and MC), the mature techniques for MDP can be easily transplanted to CFMDP, which is what we have done in our algorithm. Besides, CFMDP can be easily cast to an infinite-state MDP model.

\paragraph{Pre-Expectation Calculus}
The pre-expectation calculus and the associated expectation-transformer semantics mainly by Morgan and McIver~
\cite{mciver2005abstraction,operational_semantics_of_pgcl}
is an elegant probabilistic counterpart of the famous predicate-transformer semantics by Dijkstra in his seminal \cite{dijkstra1975guarded}.
The semantics is designed to compute expectations, which can also be made use of to compute violation probability elegantly via expectation \cite{mciver2005abstraction}.
However, the key obstacle of its automation is the same as in predicate-transformer semantics -- about synthesise a (non-)probabilistic invariant, while our method is able to run automatically.

\section{Conclusion}
\label{sec: conlusion}

% !TEX root = ./main.tex

We present a new proof rule \textsc{ProbTA} that extends trace abstraction with probability in the style of control-flow randomness, and prove its soundness and completeness.
We propose two new probabilistic models CFMDP and CFMC which extend probabilistic models with control-flow information, which allow interaction between probability theory and program verification. We present an automated algorithm to apply our proof rule.
To conclude, we provide a theoretically comprehensive framework of probabilistic trace abstraction in the control-flow randomness style, which we believe could be extended and applied in future work.

%% Acknowledgments
%\begin{acks}                            %% acks environment is optional
%                                        %% contents suppressed with 'anonymous'
%  %% Commands \grantsponsor{<sponsorID>}{<name>}{<url>} and
%  %% \grantnum[<url>]{<sponsorID>}{<number>} should be used to
%  %% acknowledge financial support and will be used by metadata
%  %% extraction tools.
%  This material is based upon work supported by the
%  \grantsponsor{GS100000001}{National Science
%    Foundation}{http://dx.doi.org/10.13039/100000001} under Grant
%  No.~\grantnum{GS100000001}{nnnnnnn} and Grant
%  No.~\grantnum{GS100000001}{mmmmmmm}.  Any opinions, findings, and
%  conclusions or recommendations expressed in this material are those
%  of the author and do not necessarily reflect the views of the
%  National Science Foundation.
%\end{acks}

%% Bibliography
\bibliography{bibs}


\begin{thebibliography}{40}

%%% the Web address.

\ifx \showCODEN    \undefined \def \showCODEN     #1{\unskip}     \fi
\ifx \showDOI      \undefined \def \showDOI       #1{#1}\fi
\ifx \showISBNx    \undefined \def \showISBNx     #1{\unskip}     \fi
\ifx \showISBNxiii \undefined \def \showISBNxiii  #1{\unskip}     \fi
\ifx \showISSN     \undefined \def \showISSN      #1{\unskip}     \fi
\ifx \showLCCN     \undefined \def \showLCCN      #1{\unskip}     \fi
\ifx \shownote     \undefined \def \shownote      #1{#1}          \fi
\ifx \showarticletitle \undefined \def \showarticletitle #1{#1}   \fi
\ifx \showURL      \undefined \def \showURL       {\relax}        \fi
% invisible to TeX
\providecommand\bibfield[2]{#2}
\providecommand\bibinfo[2]{#2}
\providecommand\natexlab[1]{#1}
\providecommand\showeprint[2][]{arXiv:#2}

\bibitem[\protect\citeauthoryear{Br{\'{a}}zdil, Esparza, Kiefer, and
  Kucera}{Br{\'{a}}zdil et~al\mbox{.}}{2013}]%
        {BEKK13}
\bibfield{author}{\bibinfo{person}{Tom{\'{a}}s Br{\'{a}}zdil},
  \bibinfo{person}{Javier Esparza}, \bibinfo{person}{Stefan Kiefer}, {and}
  \bibinfo{person}{Anton{\'{\i}}n Kucera}.} \bibinfo{year}{2013}\natexlab{}.
\newblock \showarticletitle{Analyzing probabilistic pushdown automata}.
\newblock \bibinfo{journal}{\emph{Formal Methods Syst. Des.}}
  \bibinfo{volume}{43}, \bibinfo{number}{2} (\bibinfo{year}{2013}),
  \bibinfo{pages}{124--163}.
\newblock


\bibitem[\protect\citeauthoryear{Carpenter, Gelman, Hoffman, Lee, Goodrich,
  Betancourt, Brubaker, Guo, Li, and Riddell}{Carpenter et~al\mbox{.}}{2017}]%
        {ppl_stan}
\bibfield{author}{\bibinfo{person}{Bob Carpenter}, \bibinfo{person}{Andrew
  Gelman}, \bibinfo{person}{Matthew~D Hoffman}, \bibinfo{person}{Daniel Lee},
  \bibinfo{person}{Ben Goodrich}, \bibinfo{person}{Michael Betancourt},
  \bibinfo{person}{Marcus Brubaker}, \bibinfo{person}{Jiqiang Guo},
  \bibinfo{person}{Peter Li}, {and} \bibinfo{person}{Allen Riddell}.}
  \bibinfo{year}{2017}\natexlab{}.
\newblock \showarticletitle{Stan: A probabilistic programming language}.
\newblock \bibinfo{journal}{\emph{Journal of statistical software}}
  \bibinfo{volume}{76}, \bibinfo{number}{1} (\bibinfo{year}{2017}),
  \bibinfo{pages}{1--32}.
\newblock


\bibitem[\protect\citeauthoryear{Chen and He}{Chen and He}{2020}]%
        {ast_by_omega_decompose}
\bibfield{author}{\bibinfo{person}{Jianhui Chen} {and} \bibinfo{person}{Fei
  He}.} \bibinfo{year}{2020}\natexlab{}.
\newblock \showarticletitle{Proving almost-sure termination by omega-regular
  decomposition}. In \bibinfo{booktitle}{\emph{Proceedings of the 41st ACM
  SIGPLAN Conference on Programming Language Design and Implementation}}.
  \bibinfo{pages}{869--882}.
\newblock


\bibitem[\protect\citeauthoryear{Christ, Hoenicke, and Nutz}{Christ
  et~al\mbox{.}}{2012}]%
        {smt_interpol}
\bibfield{author}{\bibinfo{person}{J{\"u}rgen Christ}, \bibinfo{person}{Jochen
  Hoenicke}, {and} \bibinfo{person}{Alexander Nutz}.}
  \bibinfo{year}{2012}\natexlab{}.
\newblock \showarticletitle{SMTInterpol: An interpolating SMT solver}. In
  \bibinfo{booktitle}{\emph{International SPIN Workshop on Model Checking of
  Software}}. Springer, \bibinfo{pages}{248--254}.
\newblock


\bibitem[\protect\citeauthoryear{Clarke, Grumberg, Jha, Lu, and Veith}{Clarke
  et~al\mbox{.}}{2000}]%
        {cegar}
\bibfield{author}{\bibinfo{person}{Edmund Clarke}, \bibinfo{person}{Orna
  Grumberg}, \bibinfo{person}{Somesh Jha}, \bibinfo{person}{Yuan Lu}, {and}
  \bibinfo{person}{Helmut Veith}.} \bibinfo{year}{2000}\natexlab{}.
\newblock \showarticletitle{Counterexample-guided abstraction refinement}. In
  \bibinfo{booktitle}{\emph{International Conference on Computer Aided
  Verification}}. Springer, \bibinfo{pages}{154--169}.
\newblock


\bibitem[\protect\citeauthoryear{Cousot and Monerau}{Cousot and
  Monerau}{2012}]%
        {prob_abstract_interpreation}
\bibfield{author}{\bibinfo{person}{Patrick Cousot} {and}
  \bibinfo{person}{Michael Monerau}.} \bibinfo{year}{2012}\natexlab{}.
\newblock \showarticletitle{Probabilistic abstract interpretation}. In
  \bibinfo{booktitle}{\emph{European Symposium on Programming}}. Springer,
  \bibinfo{pages}{169--193}.
\newblock


\bibitem[\protect\citeauthoryear{Dijkstra}{Dijkstra}{1975}]%
        {dijkstra1975guarded}
\bibfield{author}{\bibinfo{person}{Edsger~W Dijkstra}.}
  \bibinfo{year}{1975}\natexlab{}.
\newblock \showarticletitle{Guarded commands, nondeterminacy and formal
  derivation of programs}.
\newblock \bibinfo{journal}{\emph{Commun. ACM}} \bibinfo{volume}{18},
  \bibinfo{number}{8} (\bibinfo{year}{1975}), \bibinfo{pages}{453--457}.
\newblock


\bibitem[\protect\citeauthoryear{Duflot, Kwiatkowska, Norman, and
  Parker}{Duflot et~al\mbox{.}}{2004}]%
        {prism_application_protocol1}
\bibfield{author}{\bibinfo{person}{M. Duflot}, \bibinfo{person}{M.
  Kwiatkowska}, \bibinfo{person}{G. Norman}, {and} \bibinfo{person}{D.
  Parker}.} \bibinfo{year}{2004}\natexlab{}.
\newblock \showarticletitle{A Formal Analysis of {Bluetooth} Device Discovery}.
  In \bibinfo{booktitle}{\emph{Proc. 1st International Symposium on Leveraging
  Applications of Formal Methods (ISOLA'04)}}.
\newblock


\bibitem[\protect\citeauthoryear{Duflot, Kwiatkowska, Norman, and
  Parker}{Duflot et~al\mbox{.}}{2006}]%
        {prism_application_protocol2}
\bibfield{author}{\bibinfo{person}{M. Duflot}, \bibinfo{person}{M.
  Kwiatkowska}, \bibinfo{person}{G. Norman}, {and} \bibinfo{person}{D.
  Parker}.} \bibinfo{year}{2006}\natexlab{}.
\newblock \showarticletitle{A Formal Analysis of {Bluetooth} Device Discovery}.
\newblock \bibinfo{journal}{\emph{Int. Journal on Software Tools for Technology
  Transfer}} \bibinfo{volume}{8}, \bibinfo{number}{6} (\bibinfo{year}{2006}),
  \bibinfo{pages}{621--632}.
\newblock


\bibitem[\protect\citeauthoryear{Esparza, Kucera, and Mayr}{Esparza
  et~al\mbox{.}}{2004}]%
        {ppda}
\bibfield{author}{\bibinfo{person}{Javier Esparza}, \bibinfo{person}{Antonin
  Kucera}, {and} \bibinfo{person}{Richard Mayr}.}
  \bibinfo{year}{2004}\natexlab{}.
\newblock \showarticletitle{Model checking probabilistic pushdown automata}. In
  \bibinfo{booktitle}{\emph{Proceedings of the 19th Annual IEEE Symposium on
  Logic in Computer Science, 2004.}} IEEE, \bibinfo{pages}{12--21}.
\newblock


\bibitem[\protect\citeauthoryear{Etessami and Yannakakis}{Etessami and
  Yannakakis}{2005}]%
        {rmc_init}
\bibfield{author}{\bibinfo{person}{Kousha Etessami} {and}
  \bibinfo{person}{Mihalis Yannakakis}.} \bibinfo{year}{2005}\natexlab{}.
\newblock \showarticletitle{Recursive Markov chains, stochastic grammars, and
  monotone systems of nonlinear equations}. In \bibinfo{booktitle}{\emph{Annual
  Symposium on Theoretical Aspects of Computer Science}}. Springer,
  \bibinfo{pages}{340--352}.
\newblock


\bibitem[\protect\citeauthoryear{Etessami and Yannakakis}{Etessami and
  Yannakakis}{2009}]%
        {rmc}
\bibfield{author}{\bibinfo{person}{Kousha Etessami} {and}
  \bibinfo{person}{Mihalis Yannakakis}.} \bibinfo{year}{2009}\natexlab{}.
\newblock \showarticletitle{Recursive Markov chains, stochastic grammars, and
  monotone systems of nonlinear equations}.
\newblock \bibinfo{journal}{\emph{Journal of the ACM (JACM)}}
  \bibinfo{volume}{56}, \bibinfo{number}{1} (\bibinfo{year}{2009}),
  \bibinfo{pages}{1--66}.
\newblock


\bibitem[\protect\citeauthoryear{Etessami and Yannakakis}{Etessami and
  Yannakakis}{2012}]%
        {ltlmc_rmc}
\bibfield{author}{\bibinfo{person}{Kousha Etessami} {and}
  \bibinfo{person}{Mihalis Yannakakis}.} \bibinfo{year}{2012}\natexlab{}.
\newblock \showarticletitle{Model checking of recursive probabilistic systems}.
\newblock \bibinfo{journal}{\emph{ACM Transactions on Computational Logic
  (TOCL)}} \bibinfo{volume}{13}, \bibinfo{number}{2} (\bibinfo{year}{2012}),
  \bibinfo{pages}{1--40}.
\newblock


\bibitem[\protect\citeauthoryear{Goldwasser and Micali}{Goldwasser and
  Micali}{1984}]%
        {encrption_using_probability}
\bibfield{author}{\bibinfo{person}{Shafi Goldwasser} {and}
  \bibinfo{person}{Silvio Micali}.} \bibinfo{year}{1984}\natexlab{}.
\newblock \showarticletitle{Probabilistic encryption}.
\newblock \bibinfo{journal}{\emph{Journal of computer and system sciences}}
  \bibinfo{volume}{28}, \bibinfo{number}{2} (\bibinfo{year}{1984}),
  \bibinfo{pages}{270--299}.
\newblock


\bibitem[\protect\citeauthoryear{Goodman, Mansinghka, Roy, Bonawitz, and
  Tenenbaum}{Goodman et~al\mbox{.}}{2012}]%
        {ppl_church}
\bibfield{author}{\bibinfo{person}{Noah Goodman}, \bibinfo{person}{Vikash
  Mansinghka}, \bibinfo{person}{Daniel~M Roy}, \bibinfo{person}{Keith
  Bonawitz}, {and} \bibinfo{person}{Joshua~B Tenenbaum}.}
  \bibinfo{year}{2012}\natexlab{}.
\newblock \showarticletitle{Church: a language for generative models}.
\newblock \bibinfo{journal}{\emph{arXiv preprint arXiv:1206.3255}}
  (\bibinfo{year}{2012}).
\newblock


\bibitem[\protect\citeauthoryear{Grellois, Lago, and Kobayashi}{Grellois
  et~al\mbox{.}}{2020}]%
        {phors_long}
\bibfield{author}{\bibinfo{person}{Charles Grellois}, \bibinfo{person}{Ugo~Dal
  Lago}, {and} \bibinfo{person}{Naoki Kobayashi}.}
  \bibinfo{year}{2020}\natexlab{}.
\newblock \showarticletitle{On the Termination Problem for Probabilistic
  Higher-Order Recursive Programs}.
\newblock \bibinfo{journal}{\emph{Logical Methods in Computer Science}}
  \bibinfo{volume}{16} (\bibinfo{year}{2020}).
\newblock


\bibitem[\protect\citeauthoryear{Gretz, Katoen, and McIver}{Gretz
  et~al\mbox{.}}{2014}]%
        {operational_semantics_of_pgcl}
\bibfield{author}{\bibinfo{person}{Friedrich Gretz},
  \bibinfo{person}{Joost-Pieter Katoen}, {and} \bibinfo{person}{Annabelle
  McIver}.} \bibinfo{year}{2014}\natexlab{}.
\newblock \showarticletitle{Operational versus weakest pre-expectation
  semantics for the probabilistic guarded command language}.
\newblock \bibinfo{journal}{\emph{Performance Evaluation}}
  \bibinfo{volume}{73} (\bibinfo{year}{2014}), \bibinfo{pages}{110--132}.
\newblock


\bibitem[\protect\citeauthoryear{Han, Katoen, and Berteun}{Han
  et~al\mbox{.}}{2009}]%
        {trace_enumeration}
\bibfield{author}{\bibinfo{person}{Tingting Han}, \bibinfo{person}{Joost-Pieter
  Katoen}, {and} \bibinfo{person}{Damman Berteun}.}
  \bibinfo{year}{2009}\natexlab{}.
\newblock \showarticletitle{Counterexample generation in probabilistic model
  checking}.
\newblock \bibinfo{journal}{\emph{IEEE transactions on software engineering}}
  \bibinfo{volume}{35}, \bibinfo{number}{2} (\bibinfo{year}{2009}),
  \bibinfo{pages}{241--257}.
\newblock


\bibitem[\protect\citeauthoryear{He and Han}{He and Han}{2020}]%
        {trace_abstraction_termination_incremental}
\bibfield{author}{\bibinfo{person}{Fei He} {and} \bibinfo{person}{Jitao Han}.}
  \bibinfo{year}{2020}\natexlab{}.
\newblock \showarticletitle{Termination Analysis for Evolving Programs: An
  Incremental Approach by Reusing Certified Modules}.
\newblock \bibinfo{journal}{\emph{Proc. ACM Program. Lang.}}
  \bibinfo{volume}{4}, \bibinfo{number}{OOPSLA}, Article
  \bibinfo{articleno}{199} (\bibinfo{date}{nov} \bibinfo{year}{2020}),
  \bibinfo{numpages}{27}~pages.
\newblock
\urldef\tempurl%
\url{https://doi.org/10.1145/3428267}
\showDOI{\tempurl}


\bibitem[\protect\citeauthoryear{Heath, Kwiatkowska, Norman, Parker, and
  Tymchyshyn}{Heath et~al\mbox{.}}{2008}]%
        {prism_application_biology}
\bibfield{author}{\bibinfo{person}{J. Heath}, \bibinfo{person}{M. Kwiatkowska},
  \bibinfo{person}{G. Norman}, \bibinfo{person}{D. Parker}, {and}
  \bibinfo{person}{O. Tymchyshyn}.} \bibinfo{year}{2008}\natexlab{}.
\newblock \showarticletitle{Probabilistic model checking of complex biological
  pathways}.
\newblock \bibinfo{journal}{\emph{Theoretical Computer Science}}
  \bibinfo{volume}{319}, \bibinfo{number}{3} (\bibinfo{year}{2008}),
  \bibinfo{pages}{239--257}.
\newblock


\bibitem[\protect\citeauthoryear{Heizmann}{Heizmann}{[n.d.]}]%
        {site_trace_abstraction}
\bibfield{author}{\bibinfo{person}{Matthias Heizmann}.}
  \bibinfo{year}{[n.d.]}\natexlab{}.
\newblock \bibinfo{title}{Uni-Freiburg : SWT - Ultimate}.
\newblock
\newblock
\urldef\tempurl%
\url{https://monteverdi.informatik.uni-freiburg.de/tomcat/Website/?ui=tool&tool=automizer}
\showURL{%
\tempurl}
\newblock
\shownote{November 02, 2021.}


\bibitem[\protect\citeauthoryear{Heizmann, Hoenicke, and Podelski}{Heizmann
  et~al\mbox{.}}{2009}]%
        {trace_abstraction}
\bibfield{author}{\bibinfo{person}{Matthias Heizmann}, \bibinfo{person}{Jochen
  Hoenicke}, {and} \bibinfo{person}{Andreas Podelski}.}
  \bibinfo{year}{2009}\natexlab{}.
\newblock \showarticletitle{Refinement of trace abstraction}. In
  \bibinfo{booktitle}{\emph{International Static Analysis Symposium}}.
  Springer, \bibinfo{pages}{69--85}.
\newblock


\bibitem[\protect\citeauthoryear{Heizmann, Hoenicke, and Podelski}{Heizmann
  et~al\mbox{.}}{2014}]%
        {trace_abstraction_termination_cav14}
\bibfield{author}{\bibinfo{person}{Matthias Heizmann}, \bibinfo{person}{Jochen
  Hoenicke}, {and} \bibinfo{person}{Andreas Podelski}.}
  \bibinfo{year}{2014}\natexlab{}.
\newblock \showarticletitle{Termination Analysis by Learning Terminating
  Programs}. In \bibinfo{booktitle}{\emph{Computer Aided Verification}},
  \bibfield{editor}{\bibinfo{person}{Armin Biere} {and}
  \bibinfo{person}{Roderick Bloem}} (Eds.). \bibinfo{publisher}{Springer
  International Publishing}, \bibinfo{address}{Cham},
  \bibinfo{pages}{797--813}.
\newblock
\showISBNx{978-3-319-08867-9}


\bibitem[\protect\citeauthoryear{Hermanns, Wachter, and Zhang}{Hermanns
  et~al\mbox{.}}{2008}]%
        {pcegar}
\bibfield{author}{\bibinfo{person}{Holger Hermanns}, \bibinfo{person}{Bj{\"o}rn
  Wachter}, {and} \bibinfo{person}{Lijun Zhang}.}
  \bibinfo{year}{2008}\natexlab{}.
\newblock \showarticletitle{Probabilistic cegar}. In
  \bibinfo{booktitle}{\emph{International Conference on Computer Aided
  Verification}}. Springer, \bibinfo{pages}{162--175}.
\newblock


\bibitem[\protect\citeauthoryear{Kallmeyer and Maier}{Kallmeyer and
  Maier}{2013}]%
        {plcfrs}
\bibfield{author}{\bibinfo{person}{Laura Kallmeyer} {and}
  \bibinfo{person}{Wolfgang Maier}.} \bibinfo{year}{2013}\natexlab{}.
\newblock \showarticletitle{Data-driven parsing using probabilistic linear
  context-free rewriting systems}.
\newblock \bibinfo{journal}{\emph{Computational Linguistics}}
  \bibinfo{volume}{39}, \bibinfo{number}{1} (\bibinfo{year}{2013}),
  \bibinfo{pages}{87--119}.
\newblock


\bibitem[\protect\citeauthoryear{Kobayashi, Dal~Lago, and Grellois}{Kobayashi
  et~al\mbox{.}}{2019}]%
        {phors}
\bibfield{author}{\bibinfo{person}{Naoki Kobayashi}, \bibinfo{person}{Ugo
  Dal~Lago}, {and} \bibinfo{person}{Charles Grellois}.}
  \bibinfo{year}{2019}\natexlab{}.
\newblock \showarticletitle{On the termination problem for probabilistic
  higher-order recursive programs}. In \bibinfo{booktitle}{\emph{2019 34th
  Annual ACM/IEEE Symposium on Logic in Computer Science (LICS)}}. IEEE,
  \bibinfo{pages}{1--14}.
\newblock


\bibitem[\protect\citeauthoryear{Kwiatkowska, Norman, and Parker}{Kwiatkowska
  et~al\mbox{.}}{2002}]%
        {prism}
\bibfield{author}{\bibinfo{person}{Marta Kwiatkowska}, \bibinfo{person}{Gethin
  Norman}, {and} \bibinfo{person}{David Parker}.}
  \bibinfo{year}{2002}\natexlab{}.
\newblock \showarticletitle{PRISM: Probabilistic symbolic model checker}. In
  \bibinfo{booktitle}{\emph{International Conference on Modelling Techniques
  and Tools for Computer Performance Evaluation}}. Springer,
  \bibinfo{pages}{200--204}.
\newblock


\bibitem[\protect\citeauthoryear{Kwiatkowska, Norman, and Parker}{Kwiatkowska
  et~al\mbox{.}}{2012}]%
        {prism_application_algorithm}
\bibfield{author}{\bibinfo{person}{M. Kwiatkowska}, \bibinfo{person}{G.
  Norman}, {and} \bibinfo{person}{D. Parker}.} \bibinfo{year}{2012}\natexlab{}.
\newblock \showarticletitle{Probabilistic Verification of Herman's
  Self-Stabilisation Algorithm}.
\newblock \bibinfo{journal}{\emph{Formal Aspects of Computing}}
  \bibinfo{volume}{24}, \bibinfo{number}{4} (\bibinfo{year}{2012}),
  \bibinfo{pages}{661--670}.
\newblock


\bibitem[\protect\citeauthoryear{McIver, Morgan, and Morgan}{McIver
  et~al\mbox{.}}{2005}]%
        {mciver2005abstraction}
\bibfield{author}{\bibinfo{person}{Annabelle McIver}, \bibinfo{person}{Carroll
  Morgan}, {and} \bibinfo{person}{Charles~Carroll Morgan}.}
  \bibinfo{year}{2005}\natexlab{}.
\newblock \bibinfo{booktitle}{\emph{Abstraction, refinement and proof for
  probabilistic systems}}.
\newblock \bibinfo{publisher}{Springer Science \& Business Media}.
\newblock


\bibitem[\protect\citeauthoryear{McMillan}{McMillan}{2006}]%
        {craig_interpolation}
\bibfield{author}{\bibinfo{person}{Kenneth~L McMillan}.}
  \bibinfo{year}{2006}\natexlab{}.
\newblock \showarticletitle{Lazy abstraction with interpolants}. In
  \bibinfo{booktitle}{\emph{International Conference on Computer Aided
  Verification}}. Springer, \bibinfo{pages}{123--136}.
\newblock


\bibitem[\protect\citeauthoryear{Mitzenmacher and Upfal}{Mitzenmacher and
  Upfal}{2017}]%
        {efficiency_of_probabilistic_algorithm_1}
\bibfield{author}{\bibinfo{person}{Michael Mitzenmacher} {and}
  \bibinfo{person}{Eli Upfal}.} \bibinfo{year}{2017}\natexlab{}.
\newblock \bibinfo{booktitle}{\emph{Probability and computing: Randomization
  and probabilistic techniques in algorithms and data analysis}}.
\newblock \bibinfo{publisher}{Cambridge university press}.
\newblock


\bibitem[\protect\citeauthoryear{Motwani and Raghavan}{Motwani and
  Raghavan}{1995}]%
        {efficiency_of_probabilistic_algorithm_2}
\bibfield{author}{\bibinfo{person}{Rajeev Motwani} {and}
  \bibinfo{person}{Prabhakar Raghavan}.} \bibinfo{year}{1995}\natexlab{}.
\newblock \bibinfo{booktitle}{\emph{Randomized algorithms}}.
\newblock \bibinfo{publisher}{Cambridge university press}.
\newblock


\bibitem[\protect\citeauthoryear{Norman, Parker, Kwiatkowska, Shukla, and
  Gupta}{Norman et~al\mbox{.}}{2003}]%
        {prism_application_industry}
\bibfield{author}{\bibinfo{person}{G. Norman}, \bibinfo{person}{D. Parker},
  \bibinfo{person}{M. Kwiatkowska}, \bibinfo{person}{S. Shukla}, {and}
  \bibinfo{person}{R. Gupta}.} \bibinfo{year}{2003}\natexlab{}.
\newblock \showarticletitle{Using Probabilistic Model Checking for Dynamic
  Power Management}. In \bibinfo{booktitle}{\emph{Proc. 3rd Workshop on
  Automated Verification of Critical Systems (AVoCS'03)}}
  \emph{(\bibinfo{series}{Technical Report DSSE-TR-2003-2, University of
  Southampton})}, \bibfield{editor}{\bibinfo{person}{M.~Leuschel},
  \bibinfo{person}{S.~Gruner}, {and} \bibinfo{person}{S.~Lo Presti}} (Eds.).
  \bibinfo{pages}{202--215}.
\newblock


\bibitem[\protect\citeauthoryear{Rabin}{Rabin}{1976}]%
        {efficiency_of_probabilistic_algorithm_3}
\bibfield{author}{\bibinfo{person}{Michael~O Rabin}.}
  \bibinfo{year}{1976}\natexlab{}.
\newblock \bibinfo{booktitle}{\emph{Probabilistic Algorithms Algorithms and
  Complexity: New Directions and Recent Results}}.
\newblock \bibinfo{publisher}{Academic Press New York}.
\newblock


\bibitem[\protect\citeauthoryear{Rothenberg, Dietsch, and Heizmann}{Rothenberg
  et~al\mbox{.}}{2018}]%
        {trace_abstraction_incremental}
\bibfield{author}{\bibinfo{person}{Bat-Chen Rothenberg},
  \bibinfo{person}{Daniel Dietsch}, {and} \bibinfo{person}{Matthias Heizmann}.}
  \bibinfo{year}{2018}\natexlab{}.
\newblock \showarticletitle{Incremental Verification Using Trace Abstraction}.
  In \bibinfo{booktitle}{\emph{Static Analysis}},
  \bibfield{editor}{\bibinfo{person}{Andreas Podelski}} (Ed.).
  \bibinfo{publisher}{Springer International Publishing},
  \bibinfo{address}{Cham}, \bibinfo{pages}{364--382}.
\newblock
\showISBNx{978-3-319-99725-4}


\bibitem[\protect\citeauthoryear{Santos}{Santos}{1969}]%
        {prob_turing_complete}
\bibfield{author}{\bibinfo{person}{Eugene~S Santos}.}
  \bibinfo{year}{1969}\natexlab{}.
\newblock \showarticletitle{Probabilistic Turing machines and computability}.
\newblock \bibinfo{journal}{\emph{Proceedings of the American mathematical
  Society}} \bibinfo{volume}{22}, \bibinfo{number}{3} (\bibinfo{year}{1969}),
  \bibinfo{pages}{704--710}.
\newblock


\bibitem[\protect\citeauthoryear{Smith, Hsu, and Albarghouthi}{Smith
  et~al\mbox{.}}{2019}]%
        {trace_abstraction_modulo_probability}
\bibfield{author}{\bibinfo{person}{Calvin Smith}, \bibinfo{person}{Justin Hsu},
  {and} \bibinfo{person}{Aws Albarghouthi}.} \bibinfo{year}{2019}\natexlab{}.
\newblock \showarticletitle{Trace abstraction modulo probability}.
\newblock \bibinfo{journal}{\emph{Proceedings of the ACM on Programming
  Languages}} \bibinfo{volume}{3}, \bibinfo{number}{POPL}
  (\bibinfo{year}{2019}), \bibinfo{pages}{1--31}.
\newblock


\bibitem[\protect\citeauthoryear{Tolpin, van~de Meent, Yang, and Wood}{Tolpin
  et~al\mbox{.}}{2016}]%
        {ppl_anglican}
\bibfield{author}{\bibinfo{person}{David Tolpin}, \bibinfo{person}{Jan~Willem
  van~de Meent}, \bibinfo{person}{Hongseok Yang}, {and} \bibinfo{person}{Frank
  Wood}.} \bibinfo{year}{2016}\natexlab{}.
\newblock \showarticletitle{Design and Implementation of Probabilistic
  Programming Language Anglican}.
\newblock \bibinfo{journal}{\emph{arXiv preprint arXiv:1608.05263}}
  (\bibinfo{year}{2016}).
\newblock


\bibitem[\protect\citeauthoryear{Wang, Hoffmann, and Reps}{Wang
  et~al\mbox{.}}{2018}]%
        {pmaf}
\bibfield{author}{\bibinfo{person}{Di Wang}, \bibinfo{person}{Jan Hoffmann},
  {and} \bibinfo{person}{Thomas Reps}.} \bibinfo{year}{2018}\natexlab{}.
\newblock \showarticletitle{PMAF: an algebraic framework for static analysis of
  probabilistic programs}.
\newblock \bibinfo{journal}{\emph{ACM SIGPLAN Notices}} \bibinfo{volume}{53},
  \bibinfo{number}{4} (\bibinfo{year}{2018}), \bibinfo{pages}{513--528}.
\newblock


\bibitem[\protect\citeauthoryear{Wang, Sun, Fu, Chatterjee, and
  Goharshady}{Wang et~al\mbox{.}}{2021}]%
        {exponential_analysis_of_probabilistic_program}
\bibfield{author}{\bibinfo{person}{Jinyi Wang}, \bibinfo{person}{Yican Sun},
  \bibinfo{person}{Hongfei Fu}, \bibinfo{person}{Krishnendu Chatterjee}, {and}
  \bibinfo{person}{Amir~Kafshdar Goharshady}.} \bibinfo{year}{2021}\natexlab{}.
\newblock \showarticletitle{Quantitative analysis of assertion violations in
  probabilistic programs}. In \bibinfo{booktitle}{\emph{Proceedings of the 42nd
  ACM SIGPLAN International Conference on Programming Language Design and
  Implementation}}. \bibinfo{pages}{1171--1186}.
\newblock


\end{thebibliography}

%% Appendix
\appendix

\newpage

\section{Supplement Materials}
\label{sec: proof of isomorphism}

% !TEX root = ./main.tex

\subsection{Supplement Definitions}
In this section, we formally present some definitions so to get things formally proved in the following parts.

For simplicity, we denote the \emph{non-random} labels of a label set $\Sigma$ by $\Sigma^-$, that is: $\Sigma^- := \Sigma \setminus \setdef{\dprobbranch \in \Sigma}{d \in \bs{\texttt{L}, \texttt{R}}}$.

Firstly, we formally define the set of \emph{actions} of a CFMDP. Given that the label set is $\Sigma$, the set of actions is given by $A_\Sigma := \setdef{i}{\dprobbranch \in \Sigma} \uplus \Sigma^-$.

The next formalism is \emph{finite-memory strategy}, abbreviated by just strategy below:

\begin{definition}[Finite Memory Strategy]
	A Finite-Memory Strategy with respect to a CFMDP with location set $L$ and action set $A$, is a tuple $(Q, \delta, q_0)$ where $Q$ is a \emph{finite} set called the internal states, $\delta$ is a \emph{partial} function with type $L \times Q \rightharpoonup A \times Q$ and $q_0 \in Q$ is the initial state.
\end{definition}

The application of strategies is standard, except that
because we relaxed the requirement for a full definition (that $\lprobbranch$ and $\rprobbranch$ is not strictly required to appear in pairs),
the result of application thus depends on whether the next location is defined.

\begin{definition}[Application of A Finite-Memory Strategy]
	Given a strategy $\psi (Q_\psi, \delta_\psi, q_0)$, and a CFMDP $\scriptA (L, \Sigma, \delta, \ell_0, \ell_e^\scriptA)$, the induced model of \emph{applying $\psi$ to $\scriptA$} is a CFA $\scriptA^\psi (L \times Q_\psi \uplus \bs{\ell_e}, \Sigma, \delta', (\ell_0, q_0), \ell_e)$ where:
	$\ell_e$ is a new location, and $\delta'$ is given in \cref{figure: transition of application}.
\end{definition}

By this definition, we have immediately:

\begin{theorem}
	For a strategy $\psi$ and CFMDP $\scriptA$, the resulting CFA $\scriptA^\psi$ is a CFMC.
\end{theorem}

\begin{figure*}[h]
	\begin{align*}
		\delta' :=\ & \setdef{(\ell, q) \xrightarrow{\sigma} T(\ell', q')}{\sigma \in \Sigma^-, \delta_\psi(\ell, q) = (\sigma, q'), \ell \xrightarrow{\sigma} \ell' \in \delta_\scriptA} 
		\\
		\uplus\ & \bigcup \setdef{\bs{(\ell, q) \xrightarrow{\probbranch_{(i, \lbranch)}} T(\ell_1, q'), (\ell, q) \xrightarrow{\probbranch_{(i, \rbranch)}} T(\ell_2, q')}}{\delta_\psi(\ell, q) = (i, q'), \bs{\ell \xrightarrow{\probbranch_{(i, \lbranch)}} \ell_1, \ell \xrightarrow{\probbranch_{(i, \rbranch)}} \ell_2} \subseteq \delta_\scriptA \\ (\ell_1, q') \in \dom(\psi), (\ell_2, q') \in \dom(\psi)} 
		\\ 
		\uplus\ & \setdef{(\ell, q) \xrightarrow{\probbranch_{(i, d)}} T(\ell', q')}{\delta_\psi(\ell, q) = (i, q'), \ell \xrightarrow{\probbranch_{(i, d)}} \ell' \in \delta_\scriptA \\(\exists \ell''.\ell \xrightarrow{\probbranch_{(i, \overline{d})}} \ell'' \in \delta_\scriptA \wedge (\ell'', q') \notin \dom(\delta_\psi)) \vee \forall \ell''. \ell \xrightarrow{\probbranch_{(i, \overline{d})}} \ell'' \notin \delta_\scriptA} \\
		& \text{with transform function $T$ given by: } \\
		& T(\ell, q) := \begin{cases}
		\ell_e & \ell = \ell_e \\
		(\ell, q) & o.w.
	\end{cases}
	\end{align*}
	\caption{$\delta'$ of Strategy Application}
	\label{figure: transition of application}
\end{figure*}

\subsection{Proof to \cref{theorem: subset CFMC and strategy isomorphism}}
\begin{proof}[Proof to \cref{theorem: subset CFMC and strategy isomorphism}]
	Let $\scriptA = (L_\scriptA, \Sigma, \delta_\scriptA, \ell_0^\scriptA, \ell_e^\scriptA)$ and 
	$\scriptM = (L_\scriptM, \Sigma, \delta_\scriptM, \ell_0^\scriptM,\ell_e^\scriptM)$, 
	where for simplicity, we assume $\scriptA$ and $\scriptM$ share the same label set.
	Also, as CFMC is deterministic, we write $\delta_\scriptM(\ell) = (\sigma, \ell')$ or $\delta_\scriptM(\ell) = (\probbranch_{(i, \lbranch)}, \ell') \wedge (\probbranch_{(i, \rbranch)}, \ell'')$ to denote set elements. 
	So we construct $$ \psi_\scriptM := (L_\scriptM \uplus \setdef{\ell_{(T, \ell, \ell')}}{T \in L_\scriptA \uplus \bs{\bot}, \ell, \ell' \in L_\scriptM \uplus \bs{\bot}}, \delta, \ell_0) $$ where $\bot$ is a dummy symbol that is in neither $L_\scriptA$ nor $L_\scriptM$ and $\delta$ is defined as the following: we first define for $\ell \in L_\scriptM$ a supporting function $\delta^-$ as:
	$$ 
	\delta^-(\ell_\scriptA, \ell) := \begin{cases}
		(\sigma, \ell') & \delta_\scriptM(\ell) = (\sigma, \ell'), \sigma \in \Sigma^- \\
		(i, \ell_{(T, \ell_2, \ell_3)}) & \begin{matrix}\delta_\scriptM(\ell) = 
		(\probbranch_{(i, \lbranch)}, \ell_2) \wedge (\probbranch_{(i, \rbranch)}, \ell_3), \\ 
		\ell_\scriptA \xrightarrow{\probbranch_{(i, \lbranch)}} T \in \delta_\scriptA
		\end{matrix}
		\\
		(i, \ell_{(T, \ell', \bot)}) & \begin{matrix}
			\delta_\scriptM(\ell) = (\probbranch_{(i, \lbranch)}, \ell'), \\
			\ell_\scriptA \xrightarrow{\probbranch_{(i, \lbranch)}} T \in \delta_\scriptA
		\end{matrix} \\
		(i, \ell_{(\bot, \bot, \ell')}) & 
			\delta_\scriptM(\ell) = (\probbranch_{(i, \rbranch)}, \ell')
	\end{cases}
	$$
	
	Then we define $\delta$ as:
	$$
	\delta(\ell_\scriptA, \ell) := \begin{cases}
		\delta^-(\ell_\scriptA, \ell) & \ell \in L_\scriptM \\
		\delta^-(\ell_\scriptA, \ell') & \ell = \ell_{(T, \ell', \_)}, \ell_\scriptA = T, \ell' \neq \bot \\
		\delta^-(\ell_\scriptA, \ell') & \ell = \ell_{(T, \_, \ell')}, \ell_\scriptA \neq T, \ell' \neq \bot
	\end{cases}
	$$
	
	Note that we use the next target location of $\scriptA$ to help determine the next choice for the strategy when encountering $\probbranch$, and by the assumption above that for any $i$ from a single location, $\lprobbranch$ and $\rprobbranch$ will point to different locations, this will work.
	
	Finally, an induction can be performed to show that the induced CFMC of $\scriptA^{\psi_\scriptM}$ has the same accepting traces as $\scriptM$.
	For both directions, a possible induction is performed on the length of non-stuck traces.
	Notice specially that, the node $(T, \ell_{(T, \ell_1^\scriptM, \ell)})$ in $\scriptA^{\psi_\scriptM}$ is effectively just $(T, \ell_1^\scriptM)$, and $(T, \ell_{(T', \ell, \ell_2^\scriptM)})$ where $T \neq T'$ is effectively just $(T, \ell_2^\scriptM)$.
\end{proof}

\subsection{Normalisation}
The normalisation to any CFMDP is given by the process below:
1) if the two branches all point to the same location which is their origin, say $\ell$, whose edge set except the current target two edges are called $S$, we create two new nodes $\ell_1$ and $\ell_2$ that: $\ell \xrightarrow{\lprobbranch} \ell_1$, $\ell \xrightarrow{\rprobbranch} \ell_2$ and $\ell_1 \xrightarrow{\rprobbranch} \ell_2$, $\ell_1 \xrightarrow{\lprobbranch} \ell$, $\ell_2 \xrightarrow\rprobbranch \ell$, $\ell_2 \xrightarrow\lprobbranch \ell_1$ for other out edges, $\ell_1$ and $\ell_2$ just copy from $\ell$; 
2) if two branches all point to another location other than the origin, we duplicate the target location and copy all its edges into another node. We perform this process one by one, so for every step, the number of such cases will decrease by $1$. This process will end up with a CFMDP satisfying the proposition;
3) proceed with the above process one by one, first for case 1) and then case 2), as case 1) will erase self-loops without generating new ones, and case 2) will erase non-self loops without generating self-loop and non-self loops if no self-loops exists so that this process will terminate.

\subsection{Discussion on Models}
As may have been noticed, our definition to CFA is a proper subclass of FSA -- the \emph{single-accepting-node} FSA.
So the language expressible is a proper subclass of regular languages.
Also, CFMDP requires that no out edges from the ending location, this is another proper restriction to expressivity.
However, all these do not affect our final result.
Because for any program traces, the class of single-accepting-node FSA is capable of expressing.
This is due to the famous theorem over the properties of this class -- a regular language is of the single-accepting-node class, iff for any two accepting strings $s_x$ and $s_y$, and any strings $s$ over the underlying alphabet (not restrict to the accepting ones), then $s_xs$ is an accepting string iff $s_ys$ is an accepting string.
This is clearly satisfiable by programs.
As strings are terminating ones, no further string with the prefix of a terminating string except itself is acceptable, as suggested by the word ``terminating''.
So, it makes no sense to allow CFMDP, a model mimics the requirement of program CFA to even have transitions out of the ending location which intuitively marks point of termination.
Furthermore, in effects, it does not affect computation in \cref{sec: algorithm}, as discuss below.
Firstly, the generalisation of traces $\tau$, violating or not, only produces single-accepting-node FSA.
Recall some standard facts about intersection on two DFAs that: 1) a location of the result is an accepting location, iff the two sources locations are all accepting; 2) a location of the result has a transition, iff the two sources locations all have this transition out.
And that differencing between two DFAs are just intersection with complement.
So, resulting $P$ by updating with differencing with $Q$ or $A$ just again produces single-accepting-node FSA.
Also, there will be no out edges in the ending node.
Furthermore, for $A$, although generalisation produces non-proper CFMDP, but before examining,
we require that $A$ is updated by intersecting with $P$ so to make discussion not pointless, this action also ensures a proper CFMDP.

\end{document}